\documentclass[journal,twoside,web]{ieeecolor}

\definecolor{subsectioncolor}{rgb}{0,0,0}
\usepackage{cite}
\usepackage{amsmath,amssymb,amsfonts,amsthm}
\usepackage{algorithm,algorithmic}
\usepackage{graphicx}
\usepackage{textcomp}

\usepackage[hypertexnames=false, breaklinks=true]{hyperref}
\usepackage{cleveref}
\usepackage{autonum} 

\usepackage{mathtools, stmaryrd}
\usepackage{xparse}
\DeclarePairedDelimiterX{\infdivx}[2]{(}{)}{%
  #1\;\delimsize\|\;#2%
}

\newcommand{\df}{D_{f}\infdivx}
\newcommand{\ra}{R_{\alpha}\infdivx}
\newcommand{\fdiv}[1]{D_{#1}\infdivx}

\usepackage{etoolbox}
\makeatletter
\@ifundefined{color@begingroup}%
  {\let\color@begingroup\relax
   \let\color@endgroup\relax}{}%
\def\fix@ieeecolor@hbox#1{%
  \hbox{\color@begingroup#1\color@endgroup}}
\patchcmd\@makecaption{\hbox}{\fix@ieeecolor@hbox}{}{\FAILED}
\patchcmd\@makecaption{\hbox}{\fix@ieeecolor@hbox}{}{\FAILED}

\usepackage{dsfont}
\usepackage{nicefrac}
\usepackage{enumerate}

\newtheorem{lemma}{Lemma}
\newtheorem{thm}{Theorem}
\Crefname{thm}{Theorem}{Theorems}
\newtheorem{defi}{Definition}
\newtheorem{prop}{Proposition}
\Crefname{prop}{Proposition}{Propositions}

\newtheorem{coro}{Corollary}

\DeclareMathOperator{\Ima}{Im}
\DeclareMathOperator{\Tr}{Tr}

\DeclareMathOperator{\rk}{rk}

\DeclareMathOperator*{\argmin}{arg\,min}

\usepackage{tikz}
\usepackage{xcolor}
\usetikzlibrary{angles,quotes}
\usetikzlibrary{calc}
\usetikzlibrary{shapes.symbols}
\usetikzlibrary{positioning}
\usetikzlibrary{fit}
\usetikzlibrary{arrows.meta}
\usepackage{pgfplots}
\usepgfplotslibrary{fillbetween}
\usetikzlibrary{spy}
\usetikzlibrary{plotmarks}

\usetikzlibrary{patterns}
\usetikzlibrary{shapes.misc}

\tikzset{cross/.style={cross out, draw=black, opacity=0.5, minimum size=2*(#1-\pgflinewidth), inner sep=0pt, outer sep=0pt}, cross/.default={1pt}}

\begin{document}

\title{Almost-Bayesian Quadratic Persuasion (Extended Version)}
\author{Olivier Massicot and C\'edric Langbort
\thanks{Submitted on February 5, 2024.}
\thanks{O. Massicot is with the Coordinated Science Laboratory, Urbana, IL 61801 USA (e-mail: om3@illinois.edu). }
\thanks{C. Langbort is with the Coordinated Science Laboratory, Urbana, IL 61801 USA (e-mail: langbort@illinois.edu). }}

\maketitle

\begin{abstract}

In this article, we relax the Bayesianity assumption in the now-traditional model of Bayesian Persuasion introduced by Kamenica \& Gentzkow. Unlike preexisting approaches---which have tackled the possibility of the receiver (Bob) being non-Bayesian by considering that his thought process is not Bayesian yet known to the sender (Alice), possibly up to a parameter---we let Alice merely assume that Bob behaves `almost like' a Bayesian agent, in some sense, without resorting to any specific model.

Under this assumption, we study Alice's strategy when both utilities are quadratic and the prior is isotropic. We show that, contrary to the Bayesian case, Alice's optimal response may not be linear anymore. This fact is unfortunate as linear policies remain the only ones for which the induced belief distribution is known. What is more, evaluating linear policies proves difficult except in particular cases, let alone finding an optimal one. Nonetheless, we derive bounds that prove linear policies are near-optimal and allow Alice to compute a near-optimal linear policy numerically. With this solution in hand, we show that Alice shares less information with Bob as he departs more from Bayesianity, much to his detriment.
\end{abstract}

\begin{IEEEkeywords}
Bayesian persuasion, Game theory, Communication networks, Uncertain systems
\end{IEEEkeywords}

\section{Introduction}

Over the past few years, problems related to strategic information transmission (SIT), which were originally introduced and studied in the field of Information Economics, have gained relevance and garnered interest in the decision \& control, information theory, and computer science communities as well. New applications of SIT ideas, concepts and modeling paradigms in these domains include, e.g.,  adversarial sensing and estimation \cite{farokhi2016estimation,farokhi2015quadratic,kazikli2022quadratic}, persuasive interactions between humans and autonomous agents/vehicles \cite{hebbar2020stackelberg,peng2019bayesian,le2016information} and congestion mitigation  \cite{das2017reducing,zhu2018stability,zhu2019routing,massicot2019public,massicot2021competitive,zhu2022information,ferguson2022avoiding}, while tools from these fields have made it possible to investigate richer SIT problem formulations such as communication over limited communication channels \cite{le2019persuasion, vora2020information} and algorithmic approaches \cite{dughmi2016algorithmic}.

The now canonical model of Bayesian Persuasion introduced by Kamenica \& Gentzkow \cite{Kamenica} considers two actors, one of whom, the Sender, has access to the state of the world and wants to convince the other actor, the uninformed Receiver, to take actions that benefit her. In accordance with Information and Computer Theoretic practice we will henceforth refer to the Sender as ``Alice'' and to the Receiver as ``Bob.''

The setup of \cite{Kamenica}  has two crucial features. First, Alice is assumed to commit to a signaling strategy, which makes the game she plays with Bob a Stackelberg one in which she acts as the leader, and  distinguishes it from the cheap-talk formulation of \cite{crawford1982strategic} which is concerned with perfect Bayesian equilibria. This commitment assumption essentially defines the Bayesian Persuasion framework and is present in all extensions of \cite{Kamenica}, from those considering multiple senders \cite{wang2013bayesian} and/or receivers \cite{GENTZKOW2017411,castiglioni2021multi}, to costly messages \cite{nguyen2021bayesian} and online settings \cite{castiglioni2020online,castiglioni2021multi,zu2021learning}, to the possibility of Bob  acquiring additional information \cite{matyskova2018bayesian,dworczak2022preparing,hu2021robust}. 

The second crucial element in \cite{Kamenica} is the assumption that Bob is Bayesian, i.e., that  he updates his prior into a posterior using Bayes' rule upon receiving Alice's message. This Bayesianity not only delineates the kind of situations captured by the model, but also plays a central role in enabling the computation of Alice's signaling policy. Indeed, exploiting a result of Aumann \& Maschler \cite{aumann1995repeated}, Kamenica \& Gentzkow show how to fully parametrize the set of posteriors that can be held by Bob upon receiving a message from Alice which, in turn, makes it possible to reformulate her program into a theoretically tractable form. This reformulation and, hence, Bob's Bayesianity, have been instrumental in most methods aimed at determining Alice's policy (such as, e.g.,   \cite{dughmi2016algorithmic,KamenicaLinear,matyskova2018bayesian,candogan}).

Given the importance of the specific way in which Bob is assumed to update his prior in \cite{Kamenica}, multiple recent works such as, e.g., \cite{de2022non,dworczak2022preparing,hu2021robust,alonso2016bayesian,kazikli2021optimal,kosterina2022persuasion,ball2021experimental,chen2023persuading} have tried to reconcile the framework of \cite{Kamenica} with the empirical fact (confirmed in many behavioral economics experiments such as \cite{camerer1998bounded,benjamin2019errors}) that human decision makers can and often do fail to be perfectly Bayesian, either through lack of access to a correct prior, or by accessing  or incorrectly (according to Bayes' rule) processing information.

The present work is closest in spirit to \cite{de2022non} in the sense that we directly consider Bob to be non-Bayesian, and to \cite{chen2023persuading} in that Bob is close to being Bayesian, which turns out to be equivalent to being almost best-responding in the linear-quadratic game setting of this article. In contrast with most of \cite{de2022non}, however, we do not make any explicit assumption regarding the process replacing Bayes rule. Instead, we model Bob's possible posteriors via a generic \emph{robust hypothesis}, in a manner resembling the notion of an almost-maximizing agent \cite{dixon2017surfing,chen2023persuading}. More precisely, we assume that, upon receiving Alice's message, Bob's posterior lies in a suitably defined neighborhood of the correct Bayesian posterior, regardless of the specific way in which it was computed.  In so doing, we formalize the notion of ``almost-Bayesianity'' suggested at the end of \cite{de2022non} and set ourselves apart from other models which either rely on parametric uncertainty (which assume that Bob's thought process is known to Alice, save for a set of parameters such as unknown mismatched prior \cite{alonso2016bayesian,kazikli2021optimal,kosterina2022persuasion}) or make Alice account for the fact that Bob may receive private side information, be it before \cite{hu2021robust} or after \cite{dworczak2022preparing} her message.

While we believe that this robust hypothesis approach has potential to model lack of Bayesianity in general persuasion and SIT problems, we focus on a particular linear quadratic setting in this work. This is to emphasize that the operationalization of the notion of neighborhood of posteriors held by Bob matters for the resolution of Alice's program, as well as because even this relatively simple case presents interesting non-trivial features: much like the celebrated Witsenhausen's counterexample \cite{witsenhausen1968counterexample}, it presents a ``linear-quadratic-Gaussian'' situation in which linear policies may not be  optimal. In addition, and in contrast with  Witsenhausen's counterexample, finding the optimal linear policy is itself challenging.

More precisely, we consider the specific class of Bayesian persuasion games introduced by \cite{Tamura}, which has also seen many variants and applications \cite{akyol2016information,sayin2021bayesian,nadendla2018effects,sezer2023robust,ui2020lqg,lambert2018quadratic}. In this setting, the state of the world $x$ is a random vector, Bob's action $a$ is an affine function of his estimation, and Alice receives a reward quadratic in $(x,a)$. Naturally, this is referred to as linear-preference quadratic-reward Bayesian persuasion, or quadratic persuasion to remain concise. Under these assumptions, Alice's objective is linear in the covariance of the estimate, although the set of covariances Alice can induce is unclear for general priors. When the prior $\nu$ is Gaussian, this set is simply determined by two linear matrix inequalities, as shown in \cite{Tamura}. Little is known otherwise, and in fact, even when $\nu$ is finitely supported, one must resort to a relaxation of the program, \cite{sayin2021bayesian}. We first extend the results of \cite{Tamura} to slightly richer priors, then set to study the case where Bob is almost Bayesian.

In order to set the stage for this class of problems, we first present, in \Cref{sec:cex}, a solvable example of linear-quadratic communication problem in which the receiver is not exactly Bayesian. \Cref{sec:bayesian} then presents the general problem of interest; we recall Bayesian persuasion, introduce the abstract notion of almost-Bayesian agent, and further develop quadratic persuasion. In \Cref{sec:ellipsoidal}, we provide a more concrete characterization of almost-Bayesian agents in the present context. Tractability concerns push us to adopt an ``ellipsoidal'' hypothesis to contain Bob's erroneous beliefs, under which we provide optimistic and pessimistic bounds matching up to a multiplicative ratio. \Cref{sec:solving} is dedicated to analyzing the approximate programs; we first derive important structural facts, then propose a numerical solution. \Cref{sec:illustration} first confronts our approximation bounds with two analytically solvable cases, whereas its last subsection illustrates the structural results obtained in previous sections. Finally, \Cref{sec:conclusion} discusses the significance of our results. Another article of ours \cite{massicot2023almost}, whose findings are discussed with regards to those of this article in \Cref{sec:conclusion}, is devoted to the entirely solvable scalar case.

\section{A tractable example}
\label{sec:cex}

\subsection{A simple strategic communication problem}

Let us consider the following persuasion game. The state of nature $x$ is a random variable in $\mathbb R^n$ distributed according to the standard multivariate Gaussian distribution $\mathcal N(0, I_n)$. Alice knows the realization of this random variable and wants to send a message $y$ so as to lead Bob to estimate $kx$, where $k$ is a constant real number. More precisely, if Bob estimates $\hat x = \mathbb E[x\!\mid\!y]$, Alice's associated cost is $\|\hat x - kx\|^2$. 

As is customary in Bayesian persuasion, the message $y$ is a random variable whose conditional distribution given $x$ is fixed, chosen in advance by Alice and known to Bob. In other words, Alice commits to a disclosing mechanism (a policy), this in turn allows a Bayesian agent to update his prior belief to a posterior belief. The problem Alice faces is to find the optimal policy, namely the conditional law for $y$ given $x$ that minimizes her expected cost. In all generality, this could be a challenging problem, however in this simple example, it is quite easy to derive. 

This derivation mostly relies on the specificity of the problem: Alice's reward is quadratic in Bob's action ($\hat x$), and Bob's action is affine in the estimate $\hat x$. The study of such problems is the scope of linear-preference quadratic-reward persuasion as introduced by \cite{Tamura}. In our specific example,
\begin{align}
	\mathbb E[\|\hat x-kx\|^2]
	&= \Tr\Sigma - 2k\Tr \mathbb E[\hat x x^\top] + k^2 \Tr I_n \\
	&= \Tr\Sigma - 2k\Tr \mathbb E[\mathbb E[\hat x x^\top\!\mid\!\hat x]] + k^2 n \\
	&= \Tr\Sigma - 2k\Tr \mathbb E[\hat x \mathbb E[x\!\mid\!\hat x]^\top] + k^2 n \\
	&= (1-2k)\Tr\Sigma + k^2 n,
\end{align}
where $\Sigma$ is the covariance of $\hat x$ and $\Tr$ denotes the trace, noting that
\begin{equation}
	\mathbb E[x\!\mid\!\hat x] 
	= \mathbb E[\mathbb E[x\!\mid\!y]\!\mid\!\hat x]
	= \mathbb E[\hat x\!\mid\!\hat x]
	= \hat x.
\end{equation}
In general however, the objective takes a more defined form, $\Tr(D\Sigma) + c$, where $D$ is a constant symmetric matrix and $c$ is a constant real number. 

For now, notice that $\Sigma\succeq0$ as it is the covariance of $\hat x$, and notice that $I_n-\Sigma\succeq0$ as it is the covariance of $x-\hat x$. On the other hand $\Sigma=0$ can be produced by the ``no-information policy,'' sending $y=0$ at all time, whereas $\Sigma = I_n$ results from the ``full-information policy,'' signaling $y=x$ as then $\hat x = y = x$. As a result, either $1-2k>0$, $\Sigma=0$ is the only solution, sending no information is optimal; either $1-2k=0$, this is a degenerate case where all policies yield the same reward; or $1-2k<0$, $\Sigma=I_n$ is the unique solution, achieved by the full-information policy.

This instance is in accordance with the general theory of linear-preference quadratic-reward persuasion with Gaussian priors: there always exists a noisy linear policy (i.e., $y=Ax+v$ for some matrix $A$ and $v$ an independent normal variable) that is optimal. In fact, once the mean and covariance of $x$ have been reduced to $0$ and $I_n$ respectively, one can even take $A$ orthogonal projection matrix and $v=0$ without loss of generality, we term such policies ``projective policies.'' One can wonder whether this stands when Bob is not truly Bayesian.

\subsection{When Bob is not Bayesian}

The previous derivation, and in fact linear-preference quadratic-reward persuasion, both rely on the fact that Bob is Bayesian. For the purposes of this motivating example, we may relax this assumption by simply assuming that Bob's estimate $\tilde x$ is never farther than $\epsilon>0$ from $\hat x$, and let Alice plan for the worst.

Concretely, Alice can first express her expected cost by using the towering property of expectation as 
\begin{equation}
	\mathbb E[\|\tilde x-kx\|^2]
	= \mathbb E[\mathbb E[\|\tilde x-kx\|^2\!\mid\! y]].
\end{equation}
She can then assume that Bob's erroneous estimate $\tilde x$ at each $y$ maximizes her conditional cost, namely her goal is to minimize
\begin{equation}
	\mathbb E\!\left[\max_{\tilde x \in \hat x + \epsilon \mathcal B} ~\mathbb E[\|\tilde x-kx\|^2\!\mid\!y]\right],
\end{equation}
having denoted the closed Euclidean unit-ball by $\mathcal B$. The inner maximization can be developed, noting the error $\eta = \tilde x - \hat x$,
\begin{align}
	\mathbb E[\|\tilde x-kx\|^2\!\mid\!y]
	&= \mathbb E[\|\eta + \hat x-kx\|^2\!\mid\!y] \\
	&= \|\eta\|^2 + 2(1-k)\eta^\top \hat x + \mathbb E[\|\hat x-kx\|^2\!\mid\! y].
\end{align}
The last term does not depend on $\eta$, we can take it out of the maximization and average it, it becomes the original Bayesian objective. All in all, Alice tries to minimize
\begin{equation}
	(1-2k)\Tr\Sigma + k^2 n
	+ \mathbb E\!\left[\max_{\eta \in \epsilon \mathcal B} ~\|\eta\|^2 + 2(1-k)\eta^\top \hat x\right].
\end{equation}

In this simple illustrative example (and in contrast to the general case), the nested maximum can be analytically found. Therefore, Alice seeks to minimize
\begin{equation}
\label{eq:exobj}
	(1-2k)\Tr\Sigma + k^2 n + \epsilon^2+2\epsilon|1-k|\mathbb E[\|\hat x\|].
\end{equation}
The program is now much more complicated as the objective picked up a term in the mean absolute deviation, $\mathbb E[\|\hat x\|]$. However, when $k\leq\nicefrac12$, the cost of Alice is at least as large as
\begin{equation}
	k^2 n + \epsilon^2,
\end{equation}
which is achieved by revealing no information so that $\hat x =0$. Just like in the Bayesian case, sending no information is optimal. When $k=1$, the last term in \eqref{eq:exobj} vanishes and so sending the information wholly is optimal, again just like when Bob is Bayesian. In the remainder of this section, we thus consider cases where $k>\nicefrac12$ and $k\neq1$, so that there is an antagonism between maximizing $\Tr\Sigma = \mathbb E[\|\hat x\|^2]$ and minimizing $\mathbb E[\|\hat x\|]$.

The following two subsections delve into the details of how to find the optimal linear policy, and explore ``radius-threshold policies'' as an other alternative. Together, they prove the following maybe surprising result.

\begin{lemma}
	The linear policy achieving the lowest value of \eqref{eq:exobj} (i.e., Alice's ``optimal linear policy'') is either no- or full-information, with value
	\begin{equation}
	k^2 n + \epsilon^2+
	\Big(
	(1-2k)n +2\epsilon|1-k| \,\mathbb E[\|x\|] \Big)^-,
\end{equation}
	where $(.)^- = \min(.,0)$.
	When $k>\nicefrac12$ is different than $1$ and $\epsilon$ is large enough, this amounts to $k^2n+\epsilon^2$. For all these $k$, there exists a radius-threshold policy whose value is strictly better.
\end{lemma}

In other words, even if we can find the optimal linear policy---and this is quite challenging in general---, it may not be optimal over all.

\subsection{Linear policies with noise}

It is quite difficult to envision which pairs $(\Tr\Sigma,\mathbb E[\|\hat x\|])$ Alice can produce through signaling. We can nonetheless explore noisy linear policies with a certain ease. When $y=Ax+v$ is sent, where $A$ is a matrix and $v$ an independent normal random variable, $\hat x$ is normal as well so
\begin{equation}
	\mathbb E[\|\hat x\|] = \mathbb E[\sqrt{z^\top \Sigma z}],
\end{equation}
where $z\sim\mathcal N(0,I_n)$ is a dummy standard variable. As all covariances $0 \preceq \Sigma \preceq I_n$ can be produced with such noisy linear policies \cite{Tamura}, the program of Alice can be written entirely in terms of $\Sigma$. In other words, after dropping the constant terms, she is interested in solving
\begin{equation}
\label{eq:sigmaeq}
	\min_{0 \preceq \Sigma \preceq I_n} ~(1-2k)\Tr\Sigma +2\epsilon|1-k|\mathbb E[\sqrt{z^\top \Sigma z}].
\end{equation}
Since the objective is strictly concave in $\Sigma$, solutions are all extreme points of the constraint set, namely they are orthogonal projection matrices. Moreover, the objective is invariant by rotation (namely $\Sigma$ and $O\Sigma O^\top$ have the same value when $O$ is orthogonal), thus the objective value at an extreme point depends only on its rank $r$. After further inspection, the objective is concave in $r$, thus the solution is either $\Sigma = 0$ (the no-information policy), or $\Sigma = I_n$ (the full-information policy). Plugging values corresponding to both policies in \eqref{eq:sigmaeq} shows that when $\epsilon$ is large enough, Alice chooses to not disclose any information.

The lowest cost Alice can get with linear policies is thus
\begin{equation}
	k^2 n + \epsilon^2+
	\Big(
	(1-2k)n +2\epsilon|1-k| \,\mathbb E[\|x\|] \Big)^-.
\end{equation}
At fixed $k$, when $\epsilon$ is large enough the expression in brackets is positive and so her optimal linear cost becomes $k^2n + \epsilon^2$.

\subsection{An outperforming radius-threshold policy}

Alice could consider another type of message: she fixes a radius threshold $R>0$ and signals $y=x$ when $\|x\|\geq R$, $y=0$ otherwise. This policy generalizes the optimal policy obtained in the entirely solvable case $n=1$ from a recent study of ours \cite{massicot2023almost}. This fact in itself shows that linear policies are not always optimal, but as it relies on a completely different set of mathematical tools, we present an elementary argument here.

Upon receiving $y\neq0$, the conditional distribution of $x$ is $\delta_y$, hence $\hat x = y$, and when $y=0$ is sent, the conditional distribution of $x$ is symmetric, hence $\hat x =0$. In any case, $y=\hat x$ and one can estimate that
\begin{align}
	R \mathbb E[\|\hat x\|] 
	&= \mathbb E[R\|x\| \mathds1_{\{\|x\|\geq R\}}] \\
	&\leq \mathbb E[\|x\|^2 \mathds1_{\{\|x\|\geq R\}}]
	= \mathbb E[\|\hat x\|^2]. 
\end{align}
Choosing $R>R^* = \nicefrac{2\epsilon|1-k|}{2k-1}$,
\begin{align}
	&(1-2k)\Tr\Sigma + k^2 n + \epsilon^2+2\epsilon|1-k|\mathbb E[\|\hat x\|] \\
	&\quad< k^2 n + \epsilon^2 +(2k-1) (R^*-R)\mathbb E[\|\hat x\|] \\
	&\quad< k^2 n + \epsilon^2.
\end{align}
Therefore, the optimal value of Alice's program without restricting it to linear policies is strictly better than $k^2 n + \epsilon^2$, which is the value of the best linear policy for all $\epsilon$ large enough. This elementary argument shows that linear policies are not optimal.

\subsection{Discussion and a preview of things to come}

In summary, there is a class of Bayesian persuasion problems, for which optimal solutions are easily computed. Moreover, these solutions have a specific form: not only are they noisy linear policies, they are projective, that is they mute some channels by projecting the state $x$ orthogonally. When the Bayesian assumption is relaxed, however, the optimal policy fails to remain linear. 

This fact may seem reminiscent of the Witsenhausen counterexample, but with the important distinction that in the current situation even computing the optimal linear strategy is challenging. Indeed, the example presented above was chosen specifically because it could be solved in closed form, and there are multiple hurdles in the general case. The inner maximization cannot be solved analytically, and yet we are to take its average over all $\hat x$, and finally optimize over all policies.

Nonetheless, in this article we strive to do just this, with few caveats. By framing the non-Bayesian term between two bounds whose ratio is at most two, we obtain a pessimistic and an optimistic program. The pessimistic program provides an upper bound that holds for all policies, linear or not, yet surprisingly is solved by a projective policy. Since Alice prepares for the worst, this is the program that she solves. This establishes that the pessimistic solution is nearly optimal. Note that this still does not imply that projective policies are optimal, merely that they are almost optimal. We also derive a lower bound valid for linear policies, yielding an optimistic program that reflects more closely the true value of linear policies.

\section{General problem of interest}
\label{sec:bayesian}

For the purposes of making this paper self-contained, we start by reviewing the basic formulation of Bayesian persuasion from \cite{Kamenica}, before introducing and justifying the almost-Bayesian framework. We also review and expand the specific linear-quadratic persuasion setting first studied in \cite{Tamura}. Specifically, we give an exact characterization of priors for which Alice's program is solved by projective policies.

\subsection{Review of Kamenica \& Gentzkow's setup}

As mentioned in the introduction, a Bayesian persuasion game consists of two players. Alice, the sender, has access to more information than Bob, the receiver, and reveals her information according to an established scheme. After receiving the message, Bob interprets it and plays an action in order to minimize his expected cost. This action defines the loss of Alice.

To fix things, consider $(\Omega, \mathcal{F}, \nu)$ a probability space, $\mathcal{A}$ an action set for Bob, $\mathcal{M}$ a message space for Alice, and $\mathcal P(\mathcal{M})$ a space of probability measures on $\mathcal M$. The loss of both receiver and sender, $u(a, \omega)$ and $v(a, \omega)$ respectively, depend on the action taken by Bob $a$ and on $\omega$, the state of the world, observed by Alice.

Alice having chosen a disclosing mechanism $\sigma\colon \Omega \to \mathcal P(\mathcal{M})$, Bob, when Bayesian, can compute his expected cost with respect to the conditional probability (the posterior belief). His action will then be
\begin{equation}
a(m) \in \argmin_{a \in \mathcal{A}} ~\mathbb{E}[u(a,\omega) \! \mid \! m].
\end{equation}
Note that this only depends on the probability law $\mathbb{P}[ . \! \mid \! m]$. To emphasize this, we denote by $\mu$ the posterior belief held by Bob. Thus, the action of Bob is actually $a(\mu)$ (if he is indifferent, we let him choose the action that is most favorable to Alice). Further denote by $\tau$ the distribution of posteriors. The expected utility of Alice is now
\begin{equation}
\mathbb{E}_\tau[ \mathbb{E}_\mu[ v(a(\mu), \omega) ] ]
= 
\mathbb{E}_\tau[ \underset{\triangleq \hat{v}(\mu)}{\underbrace{v(a(\mu), \mu) }} ],
\end{equation}
where we used the standard notation $v(.,\mu) = \mathbb E_\mu[v(.,\omega)]$.

As pointed out in \cite{Kamenica}, exploring the case where $\Omega$ is finite, it is illuminating to write Alice's program with the distribution $\tau$ of posteriors as a variable for two reasons. First, the objective depends affinely in $\tau$, second the set $\mathcal T_\nu$ of distributions of posteriors that can be generated by a policy from the prior $\nu$, is easily described, again affinely in $\tau$. Both facts have geometric consequences which bring new light to the structure of the program. At a higher level, this simply means that Alice may instead focus on $\tau$, solve
\begin{equation}
	\label{prog:Bayes}
	\min_{\tau \in \mathcal T_\nu} ~\mathbb E_\tau[\hat v(\mu)],
\end{equation}
and later retrieve $\sigma$.

Characterizing $\mathcal T_\nu$ when $\nu$ is not finitely supported is challenging, nonetheless it is worth noting that in some cases the statistics relevant for the objective that are embedded in $\tau$ can be described simply. Gentzkow and Kamenica \cite{KamenicaLinear} explore this when $\Omega = \mathbb R$, Bob's response depends only on his estimate of the state, and Alice's loss is state-independent. More relevantly to the present work, in linear-preference quadratic-reward persuasion, only the covariance of the estimate matters and in some cases their range is well-known.

\subsection{Approximate Bayesianity}
\label{sec:bayes_inf}

While \eqref{prog:Bayes} is instrumental in revealing the structure of Alice's optimal  messaging policy for some families of function $\hat{v}$, it is only available when Bob is truly Bayesian. One way in which this assumption may fail to hold is if Bob is \textit{trying} to apply Bayes rule, yet fails because, e.g., he makes computations errors in doing so, if the computation is costly, or if the representation of the posterior distributions are not accurate in the formula. Alternatively, if one thinks of this game as a stage of a repeated process in which $\sigma$ is learned over time, there might be an error in Bob's learning, resulting in the use of an erroneous $\sigma$ in (a possibly otherwise correct) Bayes' rule.

A natural question, then, is to try and characterize the posterior beliefs that Bob may hold, as a result of such errors. To this end, we consider that Bob's erroneous posterior lies within a given safety set, parametrized by the Bayesian posterior, formally
\begin{equation}
	\mu' \in \Lambda(\mu),
\end{equation}
without further specifying how $\mu'$ is generated. This idea appeared recently in the literature, for instance as a generalization of parametric models, \cite{de2022non}. One can think of $\Lambda(\mu)$ as the set of posteriors Alice finds credible. We will refer to the correspondence $\Lambda$ as Alice's \emph{robust hypothesis.}

Realizing Bob will fail to produce accurate posteriors, Alice may want to account for the worst of his possible mistakes. To do so, Alice could expect a worst-case loss for each belief $\mu$, 
\begin{equation}
	\hat v'(\mu) 
	\triangleq 
	\sup_{\mu'\in \Lambda(\mu)}	
	~v(a(\mu'), \mu).
\end{equation}
This would naturally lead to a ``classical'' Bayesian persuasion program such as \eqref{prog:Bayes}, with $\hat v'$ replacing $\hat{v}$, i.e.
\begin{equation}
	\label{prog:almost_Bayes}
	\min_{\tau\in\mathcal T_\nu}	~\mathbb{E}_\tau[\hat v'(\mu)].
\end{equation}

Alternatively, Alice could want to account for the worst of Bob's mistakes, for every realization $\omega$. This would yield a more robust program as it would capture the worst mistake of Bob for each realization of $\omega$, and not merely for each message $m$. However, we deem this approach too conservative since Bob never observes $\omega$ before taking action, and his mistakes might thus not be correlated with $\omega$ further than through the knowledge of $m$.

Our hypothesis also singularly differs from parametric uncertainty, where Bob behaves in a specific coherent way, unknown to Alice. In this case, she would rather account for this uncertainty at the root, and not at the belief level. Informally, if $\theta\in\Theta$ is the unknown parameter and $\hat v_\theta$ denotes the conditional utility of Alice when Bob is of type $\theta$, the program of Alice should rather be
\begin{equation}
	\min_{\tau\in\mathcal T_\nu} ~\sup_{\theta\in\Theta} ~\mathbb E[\hat v_\theta(\mu)].
\end{equation}
It is nonetheless possible to consider the perhaps overly robust program
\begin{equation}
	\min_{\tau\in\mathcal T_\nu} ~\mathbb E\!\left[\sup_{\theta\in\Theta} ~\hat v_\theta(\mu)\right]
\end{equation}
which fits in our framework. It is arguably too conservative, yet it could prove useful if more amenable to analysis than the previous approach. On this topic, we refer the interested reader to our discussion in \Cref{sec:hypotheses}.

In order to make progress in characterizing how solutions of \eqref{prog:almost_Bayes} would differ from those of \eqref{prog:Bayes}, we now consider a special setup, as introduced in \cite{Tamura}. We later relax the Bayesian hypothesis, and consider the specific case of linear-quadratic persuasion.

\subsection{Linear-quadratic persuasion}

This section reviews linear-quadratic persuasion as introduced by Tamura in \cite{Tamura}, the setting of our study (albeit with an almost-Bayesian Bob). In a general linear-quadratic persuasion game, Alice observes the state of nature $x\in\mathbb R^n$ distributed according to $\nu$, a Borel probability measure on $\mathbb R^n$ centered and of covariance $I_n$ without loss of generality. She then sends a message $y\sim\sigma(x)$ with $\sigma \colon \mathbb R^n \to \mathcal P(\mathcal M)$ fixed, known by Bob and chosen by Alice. Bob then plays his best response, assumed to be affine in his estimation $\hat x = \mathbb E[x\!\mid\!y]$,
\begin{equation}
\label{eq:bobbr}
	a(\hat x) = B\hat x + b \in \mathbb R^k.
\end{equation}
Finally, Alice suffers the quadratic loss
\begin{equation}
\label{eq:alicecost}
	v(a,x) 
	= \begin{bmatrix} x\\a \end{bmatrix}^\top M \begin{bmatrix} x\\a \end{bmatrix} 
	+ p^\top \begin{bmatrix} x\\a \end{bmatrix} + q,
\end{equation}
where $M$ is symmetric and $a$ is the action played by Bob. The theoretical appeal of this model is that 1) Bob's response can be motivated as resulting from a quadratic loss as well, 2) for a given policy $\sigma$, Alice's loss only depends on the covariance of $\hat x$ as detailed in the following lemma.

\begin{lemma}[from \cite{Tamura}]
\label{lem:rewriteTamura}
	For $\sigma$ fixed, Alice's cost is
	\begin{equation}
		\mathbb E[v(a(\hat x),x)]
		= 
		\Tr(D\Sigma) + c,
	\end{equation}
	where $c$ is a constant, $D= M_{12}B+B^\top  M_{21}+B^\top  M_{22} B$ is a constant symmetric matrix, and $\Sigma$ is the covariance of $\hat x$ under policy $\sigma$.
\end{lemma}

The covariance of $\hat x$ always lies in $\mathcal S \triangleq \{\Sigma \succeq 0, ~\Sigma\preceq I_n\}$, and both bounds can be reproduced exactly with respectively no- and full-information disclosure. If we call $\mathcal S_\nu \subset \mathcal S$ the set of covariances of $\hat x$ produced by any policy, Alice's quest amounts to first finding $\Sigma$ that solves
\begin{equation}
\label{prog:Bayes_linquad_nonisotropic}
	\min_{\Sigma\in \mathcal S_\nu} ~\Tr(D\Sigma) + c,
\end{equation}
then retrieving a policy $\sigma$ that generates this covariance. At this stage, it remains unclear how to perform either step. 

The author of \cite{Tamura} notes that
\begin{equation}
\label{prog:Bayes_linquad}
\tag{BP}
	\min_{0 \preceq \Sigma \preceq I_n} 
	~\Tr(D\Sigma) + c
\end{equation}
is an upper bound on Alice's best performance (i.e., a lower bound of her lowest expected cost), and when $\mathcal S_\nu = \mathcal S$, actually equals it. Program \eqref{prog:Bayes_linquad} is immediate to solve, either numerically by recognizing it is a semi-definite program (SDP), or analytically by resorting to the following lemma, of which we will make frequent use.\footnote{This lemma is more or less already present under a different form in the proof of Theorem 1 of \cite{Tamura}, but this specific formulation is more helpful to us.}

\begin{lemma} \label{lem:trsols}
	Solutions of
	\begin{equation}
		\min_{0\preceq X\preceq I_n} ~\Tr(DX),
	\end{equation}
	are exactly all $P_D^{<0} \preceq X \preceq P_D^{\leq0}$, where $P_D^{<0},P_D^{\leq0}$ are respectively the orthogonal projection matrix on the negative and on the non-positive eigenspace of $D$ (i.e., the space spanned by the eigenvectors of $D$ associated to negative and respectively non-positive eigenvalues).
	
	The solution is unique when $P_D^{<0} = P_D^{\leq0}$ (corresponding to $D$ non-singular). This situation arises generically, however, when it is not the case, $P_D^{<0}$ is the only solution of minimal rank.
\end{lemma}

A direct consequence of this lemma, when $P^{<0}_D \in \mathcal S_\nu$, is that Alice has incentive to send some information (i.e., a policy other than no revelation) if and only if $D \not\succeq0$.

The solution of minimal rank corresponds to the situation in which Alice's policy uses the minimal number of channels while remaining optimal. It is worth noting that since \eqref{prog:Bayes_linquad} is a concave program (the objective of the minimization is concave, on a convex domain), there always is a solution that is an extreme point of the domain, thus in this case, is an orthogonal projection matrix. It is appreciable that the unique solution of minimal rank is also an orthogonal projection matrix.

\smallskip
Figuring out $\mathcal S_\nu$ for a given prior can be challenging. Nevertheless, it is possible to check whether $\mathcal S = \mathcal S_\nu$ in practice thanks to the following theorem.

\begin{thm}
\label{thm:isotrop}
	The four following statements are equivalent,
	\begin{enumerate}[(i)]
		\item $\mathcal S = \mathcal S_\nu$;
		\item Program \eqref{prog:Bayes_linquad_nonisotropic} and \eqref{prog:Bayes_linquad} have same value for all $D$;
		\item for all rotation matrix $R$, $Rx$ is distributed according to $\nu$;
		\item for all orthogonal projection matrix $P$, $\mathbb E[x\!\mid\!Px] = Px$.
	\end{enumerate}
\end{thm}

A simple yet useful byproduct of this theorem is that when $\nu$ is rotationally-invariant (i.e., condition (iii) applies) and the message is $y = Px$, with $P$ an orthogonal projection matrix, the covariance of $\hat x$ is $P$. Indeed, $\hat x = Px$ and thus $\Sigma = PI_nP^\top = P$. Hence under condition (iii), \eqref{prog:Bayes_linquad} exactly represents Alice’s program, sending $y=P^{<0}_D x$ is optimal. For this reason, we define the projective policy of (orthogonal projection) matrix P as the signaling policy $y = Px$.

Note that condition (iii) only involves rotation matrices, not all orthogonal matrices. In dimension $n\geq2$, (iii) is equivalent to the isotropy of $\nu$ (i.e., its invariance under all orthogonal transformations, see \cite{isotropic} for a study). The subtlety only occurs in the case $n=1$ for which condition (iii) is trivial, whereas the same condition allowing all orthogonal matrices further entails $\nu$ is symmetric.

This theorem uncovers the exact reasons behind the success of linear (and, in fact, projective) policies under Gaussian priors, and their failure in the counterexample provided in \cite{Tamura} with $\nu$ uniform over the unit square. As such, this theorem extends the results obtained by \cite{Tamura} for Gaussian priors to all rotationally-invariant priors, and, in fact, proves that no other prior enjoys the same properties. Given the importance of rotationally-invariant priors and the desirable properties they possess, $\nu$ will be assumed isotropic (but not necessarily Gaussian) in the remainder of this article. In this case, \eqref{prog:Bayes_linquad} is the Bayesian Program, the one Alice solves when Bob is Bayesian.

\section{Approximating Alice's program under an ellipsoidal hypothesis}
\label{sec:ellipsoidal}

When Bob is not exactly Bayesian, Alice's program is not quite as simple as explained in \Cref{sec:bayesian} since his response is not just linear in $\hat x$. In order to make progress in this case, we first need to discuss the robust hypothesis $\Lambda$ in more details. The last part of this section is dedicated to approximating the non-Bayesian term under these hypotheses.

\subsection{Tractable robust hypotheses}

We conveniently denote the average by $\bar\mu \triangleq \mathbb E_\mu[x]$ when $\mu$ is a probability measure over $\mathbb R^n$. We will also call $\Sigma_\mu = \mathbb E_\mu[(x-\bar\mu)(x-\bar\mu)^\top]$, the covariance of $x$ under belief $\mu$. In particular, $\bar\nu=0$ and $\Sigma_\nu = I_n$. 

Let $\mu'_y$ be an erroneous belief of Bob after receiving message $y$, then $\tilde x = \bar\mu'_y$ is Bob's inaccurate estimation of $x$ given $y$, whereas the Bayesian estimate is $\hat x=\bar\mu_y$. A cautious Alice tries to account for this inaccuracy. She realizes that since Bob's action only depends on $\tilde x$, she need not worry about $\mu'_y$ entirely but solely about its mean. For this reason, she only really needs to consider the set of means induced by $\Lambda(\mu)$, i.e.
\begin{equation}
	\bar\Lambda(\mu) \triangleq \{\bar\mu', ~\mu'\in\Lambda(\mu)\}.
\end{equation}
This set can take various forms depending on the specific way Bob fails to be Bayesian, according to Alice. Nonetheless, the following kind of hypotheses, in addition to lending itself to some degree of tractability, as shown later, also captures a number of `natural' ways in which Bob fails to be Bayesian, as explained in \Cref{sec:hypotheses}.

\begin{defi}
	The ellipsoidal hypothesis of parameter $C$ (and of shape $CC^\top$) is the correspondence $\bar\Lambda$ defined by
	\begin{equation}
		\bar\Lambda(\mu) = \bar\mu + C\mathcal B.
	\end{equation}
\end{defi}

Since $\bar\Lambda$ defined above only depends on $\mu$ through its mean $\bar\mu$, we henceforth will abuse notation by writing $\bar\Lambda(\bar\mu)$. Note that the ellipsoidal hypothesis of parameter $0$ is none other than the Bayesian hypothesis.

\subsection{Rewriting the program under an ellipsoidal hypothesis}

With this definition in hand, we are now in position to tackle Alice's program. To propose a sharper analysis of Alice's cost, we assume that her cost is non-negative. Under this assumption, we have the following rewriting.

\begin{lemma}
\label{lem:nonneg}
	There exist $Q\succeq0$, $l$ a vector and $r\geq0$ such that
	\begin{equation}
	\label{eq:alicecostnonneg}
		v(a(\tilde x),x)= \left(\begin{bmatrix} x\\\tilde x \end{bmatrix} - l\right)^\top Q \left(\begin{bmatrix} x\\\tilde x \end{bmatrix} - l \right) + r.
	\end{equation}
\end{lemma}

Recall that Alice's objective is $\mathbb E_\tau[\hat{v}'(\mu) ]$, where
\begin{equation}
\label{prog:hatv_lin}
	\hat{v}'(\mu) 
	=
	\sup_{\bar\mu'\in\bar\Lambda(\bar\mu)}	
	~v(a(\bar\mu'), \mu).
\end{equation}
Since \eqref{prog:hatv_lin} only depends on $\bar\mu$, the objective of Alice is only a function of the distribution $\bar\tau$ of estimates, rather than the distribution $\tau$ of the whole beliefs. Accordingly, we denote by $\bar{\mathcal T}_\nu$ the set of distributions of estimates that can be generated by a policy from the prior $\nu$. In this context, $\delta_{\bar\nu}\in\bar{\mathcal T}_\nu$ is the distribution of estimates resulting from the no-information policy, and $\nu\in\bar{\mathcal T}_\nu$ is the distribution of estimates resulting from the full-information policy. With this notation in hand, we rewrite Alice's program in the following lemma.

\begin{lemma}
\label{lem:rewrite}
	Under ellipsoidal hypothesis of parameter $C$, the program of Alice takes the form
\begin{equation}
	\label{prog:lin_almost_Bayes}
	\tag{ABP}
		\min_{\bar\tau\in\bar{\mathcal T}_\nu} 
		~\Tr(D\Sigma) +c+ \mathbb{E}_{\bar\tau}\!\left[\max_{\eta\in C\mathcal B} 
		~w(\eta, \bar\mu)\right],
\end{equation}
where explicitly
	\begin{equation}
		w(\eta,\bar\mu) = 2((Q_{21}+Q_{22})\bar\mu-Q_{21}l_1-Q_{22}l_2)^\top \eta +  \eta^\top Q_{22} \eta.
	\end{equation}
The term $\Tr(D\Sigma) + c$, where
\begin{equation}
	D= Q_{12}+Q_{21}+Q_{22}, ~c =r + l^\top Q l + \Tr Q_{11},
\end{equation}
corresponds to the Bayesian case, as can be seen by setting $C=0$. The remaining term is the penalty induced by the imprecise knowledge of Alice over Bob's belief.
\end{lemma}

Before exploring approximations, we should mention that under the ellipsoidal hypothesis, Alice has no incentive to share information to an almost-Bayesian agent if she has none to share information to a Bayesian agent. 

\begin{thm}
\label{thm:noinf}
	When an optimal strategy is to not reveal any information to the Bayesian agent (equivalently, when $D\succeq0$), the same is true for almost-Bayesian agents. More formally put, if $\Sigma=0$ is a solution of the Bayesian program \eqref{prog:Bayes_linquad}, then $\bar\tau = \delta_{\bar\nu}$ is a solution of the Almost-Bayesian Program \eqref{prog:lin_almost_Bayes}.
\end{thm}

In general, it remains unclear how to determine whether Alice would profit at all from sending a message compared to not communicating any information. Nonetheless, there are cases for which we can certify Alice wants to communicate with Bob. Having defined
\begin{equation}
\label{eq:defs}
    \begin{aligned}
    	\bar\lambda &= \overline\lambda(C^\top Q_{22} C)\\
    	E &= 4(Q_{12}+Q_{22}) CC^\top(Q_{21}+Q_{22})\\
    	f &= 4(l_1^\top Q_{12} + l_2^\top Q_{22})CC^\top (Q_{21}l_1+Q_{22}l_2),
    \end{aligned}
\end{equation}
we prove the following.

\begin{thm}
\label{thm:inf}
 	When $C^\top Q_{22} C$ is not a scaling of the identity and
	\begin{equation}
		D \prec  -\frac{f + \Tr E}{4(\bar\lambda - \bar\lambda_2)}I_n,
	\end{equation}
	where $\bar\lambda_2$ denotes the second largest eigenvalue of $C^\top Q_{22} C$,
	$\bar\tau = \delta_{\bar\nu}$ is \emph{not} a solution of \eqref{prog:lin_almost_Bayes}, even restricting to projective policies.
\end{thm}

Note that this condition remains unchanged as $C$ is scaled homothetically. In other words, in this case, Alice would never cease to send information to Bob as he becomes less and less Bayesian, even if she is restricted to projective policies.

\subsection{The framing programs}

Solving Alice's program (ABP) exactly is challenging as even parametrizing the feasible set is non-trivial. Fortunately, we can resort to well-controlled upper and lower bounds of the objective that only depend on the covariance of the estimate, whose feasible set has been exactly characterized in Theorem 1. Before presenting them, we recall all the assumptions made so far. We have assumed that the prior $\nu$ is isotropic once centered and reduced (i.e., so that $\bar\nu=0$ and $\Sigma_\nu=I_n$), that Bob's action is affine in his estimate $\tilde x = \bar\mu'$, that
\begin{equation}
	\bar\mu' \in \bar\mu + C\mathcal B,
\end{equation}
and that Alice's loss is quadratic in $(x,a)$ and non-negative, so that incorporating the coefficient of Bob's affine best-response in her cost, it takes the form \eqref{eq:alicecostnonneg}.

With all this in place, we can now focus on bounding Alice's cost. In the first theorem, we derive a general lower and upper bound.

\begin{thm}
\label{thm:boundany}
	For any $\bar\tau\in\bar{\mathcal T}_\nu$, namely for any policy,
	\begin{align}
		c+\Tr(D\Sigma)+\bar\lambda
		&\leq \mathbb E_{\bar\tau}[\hat v'(\bar\mu)] \\
		&\leq c+\Tr(D\Sigma)+\bar\lambda+\sqrt{f+\Tr(E\Sigma)} \\
		&\leq 2(c+\Tr(D\Sigma)+\bar\lambda)
	\end{align}
	with $\bar\lambda, E, f$ as in \eqref{eq:defs}, and $\Sigma = \Sigma_{\bar\tau}$ the covariance of the estimate under $\bar\tau$. In particular, the cost of projective policies solutions of \eqref{prog:Bayes_linquad} and \eqref{prog:pessimistic} (defined below) is at most twice the optimal cost.
\end{thm}

Fortunately, the lower bound is strong enough to always match the upper bound up to a fixed ratio of $2$. This is due to the fact that, even though the penalty term may not be well approximated alone, it remains well controlled considering the contribution of the Bayesian term. Turning to the more congenial class of projective policies, we find an even stronger lower bound.

\begin{thm} 
\label{thm:boundproj}
	For any projective policy (and corresponding orthogonal projection covariance matrix $\Sigma$) and $\beta\in[0,1]$,
	\begin{equation}
	\label{eq:sandwich}
	\begin{aligned}
		&c+\Tr(D\Sigma)+(1-\beta^2)\bar\lambda+\beta\kappa\sqrt{f+\Tr(E\Sigma)}\\
		&\quad\leq \mathbb E_{\bar\tau}[\hat v'(\bar\mu)] \\
		&\quad\leq c+\Tr(D\Sigma)+\bar\lambda+\sqrt{f+\Tr(E\Sigma)} \\
		&\quad\leq \bar\gamma(c+\Tr(D\Sigma)+(1-\bar\beta^2)\bar\lambda+\bar\beta\kappa\sqrt{f+\Tr(E\Sigma)}),
	\end{aligned}
	\end{equation}
	where explicitly
	\begin{equation}
		\kappa = \frac{\mathbb E[|x_1|]}{\sqrt{1+\mathbb E[|x_1|]^2}},
		~\bar\beta = \frac{\kappa}{1+\kappa^2},
		~\bar\gamma = 1+\frac1{1+\kappa^2}.
	\end{equation}
\end{thm}

The ratio $\bar\gamma$ depends on the prior distribution, and lies between $\nicefrac53$ and $2$. For Gaussian priors, $\bar\gamma$ is independent of the dimension and approximately equals $1.72$. A more precise statement, \Cref{prop:gamma}, is formulated in \Cref{sec:gamma}.

\medskip
For future reference, we now gather all four programs of interest in one list:
\begin{enumerate}
	\item the \emph{Bayesian Program} is
	\begin{equation}
		\min_{0 \preceq \Sigma \preceq I_n} 
		~\Tr(D\Sigma) + c; \tag{\ref{prog:Bayes_linquad}}
	\end{equation}
	\item the \emph{Pessimistic Program} is
	\begin{equation}
        \label{prog:pessimistic}
        \tag{PP}
        	\min_{0\preceq\Sigma\preceq I_n}~ \Tr(D\Sigma) + c + \bar\lambda + \sqrt{f+\Tr(E\Sigma)};
        \end{equation}
        \item the \emph{Universal Optimistic Program} is
    	\begin{equation}
        \label{prog:genoptimistic}
        \tag{UOP}
        	\min_{0\preceq\Sigma\preceq I_n}~ \Tr(D\Sigma) + c + \bar\lambda;
        \end{equation}
	\item and the \emph{Projective Optimistic Program} is
	\begin{equation}
        \label{prog:optimistic}
        \tag{POP}
        \min_{0\preceq\Sigma\preceq I_n}~ \Tr(D\Sigma) + c+(1-\beta^2)\bar\lambda + \beta\kappa \sqrt{f+\Tr(E\Sigma)}.
        \end{equation}
\end{enumerate}

The Pessimistic Program and Universal Optimistic Program correspond to the upper-bound and lower-bound obtained in \Cref{thm:boundany}, respectively. \eqref{prog:genoptimistic} has the same solution as \eqref{prog:Bayes_linquad}. This implies that, as discussed in \Cref{thm:boundany}, both the Bayesian and the Pessimistic solution yield a cost at most twice the optimal one. In spite of the fact that both solutions offer the same worst-case guarantee, twice the Universal Optimistic Program (later referred to as (2UOP)) is in all generality a weaker upper bound of the True Program than the Pessimistic Program, which justifies to search for the optimal solution of \eqref{prog:pessimistic}. Finally, the Projective Optimistic Program, which is a lower bound on the cost of Alice when using projective policies, is derived from \Cref{thm:boundproj}. Although the theorem only speaks of projective policies, and hence of covariances that are extremal in $\mathcal S$, the objective of the minimization is concave, thus the constraint set can be extended to $\mathcal S$ entirely. This program is mostly identical to \eqref{prog:pessimistic}, save for the multiplicative constants that adorn the error terms. In this respect, being able to solve \eqref{prog:optimistic} amounts to being able to solve \eqref{prog:pessimistic}. For this reason, and since Alice is preparing for the worst, \eqref{prog:pessimistic} remains our main object of study.

In summary, \eqref{prog:pessimistic} is a universal lower bound on Alice's best performance, whereas we dispose of two optimistic programs, \eqref{prog:genoptimistic} and \eqref{prog:optimistic}, depending on whether we allow any policy---not just projective policies---to be implemented. The cost of a policy being pinned down up to a ratio of a half (for all priors), finding a solution to \eqref{prog:pessimistic} appears to be a good proxy for solving the true program \eqref{prog:lin_almost_Bayes}. A projective policy solves this program, and for those, we dispose of an improved bound. \eqref{prog:genoptimistic} seems to indicate that there could be better non-projective policies, however they remain inaccessible as it already proves arduous to even represent such general policies.

\section{Analysis of the Pessimistic Program}
\label{sec:solving}

This section sheds light on the Pessimistic Program \eqref{prog:pessimistic} and the structure of its solutions. It also presents a numerical method to solve it. We then verify that the structure of the numerical solutions agrees with theoretical predictions.

\subsection{Structural facts}

Much like for the Bayesian Program, there are a few things that can be said about the solutions of \eqref{prog:pessimistic}. First of all, the program is concave, so just like in the Bayesian case, there exists a solution that is an extreme point of $\mathcal S$, thus corresponding to a projective policy. 

In contrast with the Bayesian case, it may so happen that \eqref{prog:pessimistic} has multiple solutions. However, as the next proposition states, all solutions of minimal rank are orthogonal projection matrices just like in the Bayesian case. We again use rank as a proxy for the amount of information shared by Alice, since when $P$ is an orthogonal projection matrix, its rank denoted $\rk P$ corresponds to the number of active channels in the policy $y=Px$.

\begin{prop} \label{prop:orth}
	Solutions of minimal rank of \eqref{prog:pessimistic} are all orthogonal projection matrices.
\end{prop}

Having decided to use the rank of an orthogonal projection matrix as a measure of information provided by Alice, it is natural to inspect how the minimal rank of a solution varies as the hypothesis grows weaker, i.e., as $\bar\Lambda$ grows larger with respect to the inclusion order. In all generality, there may be no monotonicity. Nevertheless, it turns out that the minimal rank of a solution decreases as the hypothesis grows weaker, provided it grows homothetically.

\begin{thm} \label{thm:monotonicity}
	Let $\Sigma_1,\Sigma_2$ be solutions of minimal rank of \eqref{prog:pessimistic} under ellipsoidal hypothesis of respective shape $\epsilon_1^2 CC^\top$ and $\epsilon_2^2 CC^\top$. Then $\epsilon_1\leq\epsilon_2$ implies $\rk \Sigma_1 \geq \rk \Sigma_2$.
\end{thm}

This theorem admits a direct corollary which, in essence, states that Alice is willing to share more information to a Bayesian agent, less information to an almost-Bayesian agent when she is optimistic, and the least information when she is pessimistic. 

\begin{coro}
\label{coro:monotonicity}
	The minimal rank of a solution of \eqref{prog:Bayes_linquad} is larger than or equal to that of \eqref{prog:optimistic}, which itself is larger than or equal to that of the \eqref{prog:pessimistic}.
\end{coro}

From this corollary, we recover the structural result of \Cref{thm:noinf} about the true program \eqref{prog:lin_almost_Bayes} in all our programs.

\begin{coro}
\label{cor:noinf_exp}
	Whenever $D\succeq0$, the minimal solution of \eqref{prog:pessimistic}, \eqref{prog:genoptimistic} and \eqref{prog:optimistic} is $\Sigma=0$, corresponding to the no-information policy.
\end{coro}

\subsection{Numerical solution}
\label{sec:numerical}

As it stands, \eqref{prog:pessimistic} is not in a convenient form. It is concave, and a square-root term sits cumbersomely in the midst of the objective. We cannot hope to directly solve the program with readily available methods, however we can introduce, for $t\geq0$,
\begin{equation}
	h(t) \triangleq \min_{\substack{0\preceq X\preceq I_n\\\text{s.t. }\Tr(EX)\leq t}} ~\Tr(DX).
\end{equation}
Evaluating $h$ at a given $t$ is relatively easy, as it is a semi-definite program (SDP). If we have a fine enough understanding and estimation of $h$ available, we may resort to the following proposition.

\begin{prop}
\label{prop:line}
	$Y\in\mathcal S$ solves \eqref{prog:pessimistic} if and only if $Y$ solves the program defining $h(\Tr(EY))$, and $\Tr(EY)$ solves the program
	\begin{equation}
	\label{prog:line}
		c + \bar\lambda + \min_{t\geq0} ~h(t) + \sqrt{f+t}.
	\end{equation}
	Moreover, both \eqref{prog:pessimistic} and \eqref{prog:line} have the same value.
\end{prop}

One can thus solve \eqref{prog:line} by a simple one-dimensional grid search, then retrieve an optimal argument of \eqref{prog:pessimistic}. In actuality however, one only obtains a suboptimal solution through grid search, so the objective of \eqref{prog:line} needs to be studied in order to provide adequate guarantees as to the suboptimality. Fortunately, $h$ enjoys many desirable properties that can be used to establish those guarantees, and we have the following proposition.

\begin{prop}
\label{prop:suboptim}
	Call $\bar t = \Tr(EP_D^{<0})$. Consider $(u_n)_{0\leq n\leq N}$ an increasing sequence with $u_0 = 0$ and $u_N\geq\bar t$. Call
	\begin{equation}
		\rho = \max_{0\leq n< N} ~\sqrt{f+u_{n+1}}-\sqrt{f+u_n},
	\end{equation}
	then 
	\begin{align}
		\min_{t\geq0} ~h(t) + \sqrt{f+t}
		&\leq \min_{0\leq n\leq N} ~h(u_n) + \sqrt{f+u_n} \\
		&\leq \min_{t\geq0} ~h(t) + \sqrt{f+t} + \rho.
	\end{align}
\end{prop}

As a result, a simple strategy for finding a $\rho$-suboptimal solution consists in first cutting $[\sqrt f, \sqrt{f+\bar t}]$ into smaller intervals of length $\rho$ of the form
\begin{equation}
	[\sqrt{f+u_n},\sqrt{f+u_{n+1}}].
\end{equation}
Then $h$ is evaluated at each $u_n$, and the point yielding the lowest value $h(u_n) + \sqrt{f+u_n}$ is selected. Over all, this takes
	\begin{equation}
		\left\lceil \frac{\sqrt{f+\bar t} - \sqrt f}{\rho} \right\rceil
	\end{equation}
calls to the SDP oracle.

\subsection{Consistency of structural and numerical results}

\Cref{prop:line,prop:suboptim} provide a numerical procedure to find \emph{a} suboptimum to \eqref{prog:pessimistic}, without guaranteeing it is an orthogonal projection matrix. However, knowing \Cref{prop:orth}, it would be natural to look for solutions of \eqref{prog:pessimistic} that are orthogonal projection matrices. On top of that, all the policies we have considered thus far are projective, whose covariances must be orthogonal projection matrices.

To remedy this apparent discrepancy, consider $X^*$ a suboptimal solution to \eqref{prog:pessimistic}. By diagonalizing it, it is relatively easy to write it as a convex combination of at most $n+1$ orthogonal projection matrices:
\begin{equation}
	X^* = \sum_{i=0}^n \lambda_i X_i.
\end{equation}
Since the objective of \eqref{prog:pessimistic} is concave, some $X_i$ must perform no worse than $X^*$, this provides Alice with a suboptimal projective policy.

In practice, however, we have noted that $X^*$ is a convex combination of at most two orthogonal projection matrices. Indeed, having reduced the problem so that
\begin{equation}
	\ker D \cap \ker E = \{0\},
\end{equation}
generically $\rk(D-\lambda E) \geq n-1$ for all $\lambda>0$, and the following proposition holds. The symbol $\ker$ here denotes the null-space.

\begin{prop}
\label{prop:coherence}
	If $\ker D \cap \ker E = \{0\}$ and $\rk(D-\lambda E) \geq n-1$ for all $\lambda>0$, then for all $t\in(0,\bar t)$, the program defining $h(t)$ has a unique solution, which is a convex combination of at most two orthogonal projection matrices. 
\end{prop}

\section{Illustrations}
\label{sec:illustration}

In order to illustrate the tightness of our approximation bounds, we first compare them against two cases we can entirely solve: the unidimensional case (i.e., when $n=1$), and the opening example. Specifically, we are interested in how the Pessimistic Program solution differs from the true optimal projective policy. The last subsection numerically solves an arbitrary instance.

\subsection{The unidimensional case}

Calling $l_0 = Q_{21}l_1+Q_{22}l_2$, the actual cost of Alice under a policy yielding $\bar\tau$ is,
\begin{align}
	&r + l_0^\top Q_{22}^{-1} l_0 + l_1^\top (Q_{11}-Q_{12}Q_{22}^{-1}Q_{21}) l_1 + \Tr Q_{11} \\
	&\quad+ \Tr(D\Sigma) + \mathbb E_{\bar\tau}\!\left[\max_{\eta\in C\mathcal B} ~w(\eta,\bar\mu)\right],
\end{align}
where
\begin{equation}
	w(\eta,\bar\mu) = 2((Q_{21}+Q_{22})\bar\mu-l_0)^\top \eta +  \eta^\top Q_{22} \eta.
\end{equation}
The other pessimistic and optimistic costs follow a similar pattern, only the error terms differ but remain independent of $Q_{11}$ and $l$ (once $l_0$ is set). It only requires simple algebra to see that the ratio of the true cost of a given policy over the true optimal cost is the largest when $r$ and $Q_{11}$ are the smallest. Hence, the ratio of the true cost of the pessimistic solution over the true optimal cost is maximum when considering $r=0$ and $Q_{11} = Q_{12}Q_{22}^{-1}Q_{21}$. This also applies to the optimistic solution.

\subsubsection{Tightness of approximations}

Looking back at how we derived \Cref{thm:boundany,thm:boundproj}, the first obstacle was to solve the inner maximization of \eqref{prog:lin_almost_Bayes}. We used two lemmas to help us, \Cref{lem:sproc,lem:prettybound1}. The first lemma turns the general $n$-dimensional optimization into a unidimensional convex program, it is exact and relies on an S-procedure followed by a Schur complement (see \cite{boyd1994linear}). The second lemma approximates the value of this simpler program, so that all in all, for all $\beta\in[0,1]$,
\begin{equation}
\label{eq:together}
\begin{aligned}
	(1-\beta^2)\bar\lambda + 2\beta\mathbb E[\|v\|]
	&\leq \mathbb E\!\left[\max_{\eta\in C\mathcal B} ~w(\eta, \bar\mu)\right] \\
	&\leq \bar\lambda + 2\mathbb E[\|v\|],
\end{aligned}
\end{equation}
where,
\begin{equation}
	v = C^\top((Q_{21}+Q_{22})\bar\mu - l_0).
\end{equation}

When $n=1$ the error term can be explicitly computed as
\begin{equation}
	\mathbb E\!\left[\max_{\eta\in C\mathcal B} ~w(\eta, \bar\mu)\right]
	= \bar\lambda + 2\mathbb E[\|v\|].
\end{equation}
So these first steps towards \eqref{prog:pessimistic}---the program Alice ultimately solves---are actually exact. We still cannot provide an optimal solution in all generality, but when $n=1$ we can find the best projective policy. Indeed, there are only two such policies: full- and no-information. In the first case, $\hat x = x$ and in the second case $\hat x = 0$. 

In the no-information case, the approximation
\begin{equation}
	\mathbb E[\|v\|] \leq \sqrt{\mathbb E[\|v\|^2]},
\end{equation}
is actually exact as the distribution of $v$ is a Dirac, so the Pessimistic Program matches the reality. In the full-information case, the relation between $\mathbb E[\|v\|]$ and $\sqrt{\mathbb E[\|v\|^2]}$ is a tad more complicated. Nonetheless, for unidimensional Gaussian priors, we have the following result.
\begin{lemma}
\label{lem:gaussian1d}
When $\nu\sim\mathcal N(0,1)$, $a,b\in\mathbb R$, 
\begin{equation}
	\sqrt{\frac2\pi} \sqrt{\mathbb E[(a+bx)^2]} \leq \mathbb E[|a+bx|] \leq \sqrt{\mathbb E[(a+bx)^2]},
\end{equation}
the lower bound occurring exactly when $a=0$.
\end{lemma}

\subsubsection{Comparing Optimistic, True and Pessimistic solutions}

In the no-information case, the true objective is the same as in the Pessimistic Program. In the full-information case however, the three programs ascribe different values, which we want to compare to each other. Following the reduction mentioned at the beginning this section and rescaling the cost to obtain $Q_{22}=1$, we have 
\begin{equation}
	r=0 \text{ and } Q = \begin{bmatrix} k^2 & -k \\ -k & 1 \end{bmatrix},
\end{equation}
where $k,l_0$ are yet to be chosen. With these parameters, $D=1-2k$, thus we focus our attention on the cases where $1-2k<0$. Moreover, when $k=1$, $v$ is a constant and all programs make the same prediction, we disregard this case. In addition, as discussed in \Cref{lem:gaussian1d}, the pessimistic value for the full-information policy is the most conservative (and so the optimistic value is closer to the true value) when $l_0=0$. In the interest of showing how the Pessimistic Program performs at its worst, we study this very case. The programs then assume forms where the exact value of $k$ only changes the relative importance of the scaling of the hypothesis, $\epsilon=|C|$, we let then $k=2$.

This results in the various costs values presented in \Cref{tab:objs}. In accordance with \Cref{thm:monotonicity}, full-information is optimal for lower $\epsilon$, and no-information becomes optimal past a threshold. The threshold corresponding to the true program is
\begin{equation}
	\epsilon^* = \frac{3\sqrt{2\pi}}4
\end{equation}
while the pessimistic and \emph{projective} optimistic thresholds are respectively
\begin{equation}
		\epsilon^- = \frac32 = \underset{\approx0.80}{\underbrace{\sqrt{\frac2\pi}}} \epsilon^*, ~
		\epsilon^+ = \frac{3(4+\pi)}4 
		= \underset{\approx2.85}{\underbrace{\frac{4+\pi}{\sqrt{2\pi}}}} \epsilon^*.
\end{equation}
The fact that $\epsilon^- \leq \epsilon^+$ agrees with the prediction of \Cref{coro:monotonicity}. Indeed, when no-information is optimal for \eqref{prog:optimistic} at a given value of $\epsilon$, it also is the case for \eqref{prog:pessimistic}.

As a result, when $\epsilon \leq \epsilon^-$ or $\epsilon \geq \epsilon^+$, all strategies agree. When $\epsilon \in (\epsilon^-, \epsilon^*)$, however, the pessimistic strategy is suboptimal whereas the optimistic strategy is optimal. When $\epsilon \in (\epsilon^*, \epsilon^+)$, the opposite happens. Qualitatively, the pessimistic solution is better in the sense that the range in which it is dominated by the optimistic solution is smaller than the converse.

\begin{table}
\caption{Objective values at NI ($\Sigma=0$) and FI ($\Sigma=1$), the only two projective policies when $n=1$.}
\label{tab:objs}
\centering
\begin{tabular}{|c|c|c|}
\hline
& NI & FI \\ \hline
\eqref{prog:lin_almost_Bayes} & $4+\epsilon^2$ & $1+2\sqrt{\frac2\pi}\epsilon + \epsilon^2$ \\
\eqref{prog:pessimistic} & $4+\epsilon^2$ & $1+2\epsilon + \epsilon^2$ \\
\eqref{prog:optimistic} & $4+(1-\beta^2)\epsilon^2$ & $1+2\beta\kappa \epsilon + (1-\beta^2) \epsilon^2$ \\\hline
\end{tabular}
\end{table}

\subsubsection{Graphical comparisons}

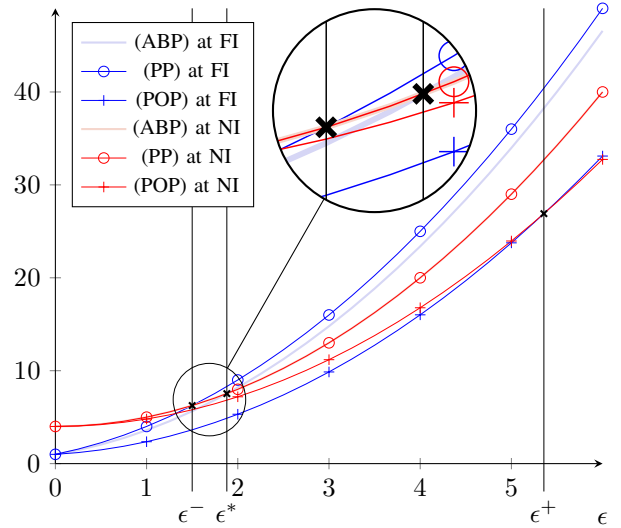
\begin{figure}
\centering
\begin{tikzpicture}[spy using outlines={circle, magnification=2.8, connect spies}]
\begin{axis}[clip=false,
					enlargelimits=false,
					axis lines=left,
		            xlabel=$\epsilon$,
					xlabel style={
					  at={(ticklabel* cs:1)},
				      anchor=north west,
				      below
				    },
					xmin=0,
					xmax=6,
					xtick={0,1,...,5},
		            ymin=0,
		            ymax=49,
		            mark repeat=4,
		            legend pos=north west,
					legend style={font=\footnotesize},
		            width=\columnwidth]
	
	\def\onemb{(12+6*pi+pi^2)/(4+pi)^2}

	\addplot[blue!80!yellow!20,thick,domain={0:6}] {1+2*sqrt(2/pi)*x+x^2};
	\addlegendentry{\eqref{prog:lin_almost_Bayes} at FI}
	\addplot[blue,very thin,mark=o,domain={0:6}] {1+2*x+x^2};
	\addlegendentry{\eqref{prog:pessimistic} at FI}
	\addplot[blue,very thin,mark=+,domain={0:6}] {1+4*x/(4+pi)+\onemb*x^2};
	\addlegendentry{\eqref{prog:optimistic} at FI}

	\addplot[red!80!green!20,thick,domain={0:6}] {4+x^2};
	\addlegendentry{\eqref{prog:lin_almost_Bayes} at NI}
	\addplot[red,very thin,mark=o,domain={0:6}] {4+x^2};
	\addlegendentry{\eqref{prog:pessimistic} at NI}
	\addplot[red,very thin,mark=+,domain={0:6}] {4+\onemb*x^2};
	\addlegendentry{\eqref{prog:optimistic} at NI}
			
	\addplot[very thin, opacity=0.5] coordinates {(3*sqrt(2*pi)/4,-3) (3*sqrt(2*pi)/4,49)};
	\addplot[very thin, opacity=0.5] coordinates {(3/2,-3) (3/2,49)};
	\addplot[very thin, opacity=0.5] coordinates {(3*(4+pi)/4,-3) (3*(4+pi)/4,49)};
	
	\node[below] at (axis cs: {3/2},-3) {$\epsilon^-$};
	\node[below] at (axis cs: {3*sqrt(2*pi)/4},-3.3) {$\epsilon^*$};
	\node[below] at (axis cs: {3*(4+pi)/4},-3) {$\epsilon^+$};

	\draw[thick, opacity=0.5] (axis cs: {3/2},{4+(3/2)^2}) node[cross=2pt, opacity=0.5]{};
	\draw[thick, opacity=0.5] (axis cs: {3*sqrt(2*pi)/4},{4+(3*sqrt(2*pi)/4)^2}) node[cross=2pt, opacity=0.5]{};
	\draw[thick, opacity=0.5] (axis cs: {3*(4+pi)/4},{4+\onemb*(3*(4+pi)/4)^2}) node[cross=2pt, opacity=0.5]{};
	
	\coordinate (spypoint) at (axis cs:1.69,6.9);
	\coordinate (magnifyglass) at (axis cs:3.5,38);
\end{axis}
	
	\spy [size=2.7cm] on (spypoint) in node[fill=white] at (magnifyglass);
\end{tikzpicture}
\caption{Plot of all the objectives.}
\label{fig:q01objs}
\end{figure}

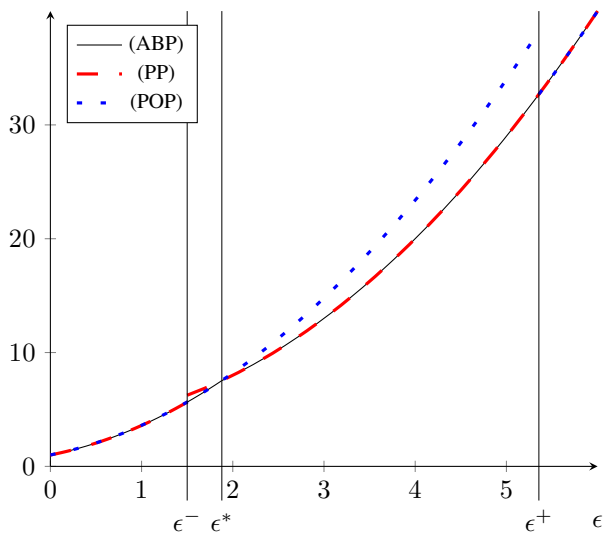
\begin{figure}
\centering
\begin{tikzpicture}
\begin{axis}[clip=false,
					enlargelimits=false,
					axis lines=left,
		            xlabel=$\epsilon$,
					xlabel style={
					  at={(ticklabel* cs:1)},
				      anchor=north west,
				      below
				    },
					xmin=0,
					xmax=6,
					xtick={0,1,...,5},
		            ymin=0,
		            ymax=40,
		            ytick={0,10,20,30},
		            mark repeat=4,
		            legend pos=north west,
					legend style={font=\footnotesize},
		            width=\columnwidth]
	
	\def\onemb{(12+6*pi+pi^2)/(4+pi)^2}

	\addplot[domain={0:3*sqrt(2*pi)/4},forget plot] {1+2*sqrt(2/pi)*x+x^2};
	\addplot[domain={3*sqrt(2*pi)/4:6}] {4+x^2};
	\addlegendentry{\eqref{prog:lin_almost_Bayes}}
	
	\addplot[very thick,dashed,dash pattern=on 8pt off 8pt,red,domain={0:3/2},forget plot] {1+2*sqrt(2/pi)*x+x^2};
	\addplot[very thick,dashed,dash pattern=on 8pt off 8pt,red,domain={3/2:6}] {4+x^2};
	\addlegendentry{\eqref{prog:pessimistic}}
	
	\addplot[very thick,dotted,dash pattern=on 2pt off 7pt,blue,domain={0:3*(4+pi)/4},forget plot] {1+2*sqrt(2/pi)*x+x^2};
	\addplot[very thick,dotted,dash pattern=on 2pt off 7pt,blue,domain={3*(4+pi)/4:6}] {4+x^2};
	\addlegendentry{\eqref{prog:optimistic}}
	
	\addplot[very thin, opacity=0.5] coordinates {(3*sqrt(2*pi)/4,-3) (3*sqrt(2*pi)/4,40)};
	\addplot[very thin, opacity=0.5] coordinates {(3/2,-3) (3/2,40)};
	\addplot[very thin, opacity=0.5] coordinates {(3*(4+pi)/4,-3) (3*(4+pi)/4,40)};
	
	\node[below] at (axis cs: {3/2},-3) {$\epsilon^-$};
	\node[below] at (axis cs: {3*sqrt(2*pi)/4},-3.3) {$\epsilon^*$};
	\node[below] at (axis cs: {3*(4+pi)/4},-3) {$\epsilon^+$};
\end{axis}
	
\end{tikzpicture}
\caption{Plot of the true cost of the optimal solutions to each program.}
\label{fig:q01vals}
\end{figure}

We plot the various objectives with $\epsilon\geq0$ varying in \Cref{fig:q01objs}. In red, we represent the values of the no-information policy and in blue, the values of the full-information policy. The thick pastel lines represent the true values, the lines with \tikz[baseline=-0.5ex]\draw (0,0) circle (2pt); marks represent the pessimistic bound, and the ones with \tikz[baseline=-0.5ex]{\draw (-0.1,0) -- (0.1,0); \draw (0,-0.1)--(0,0.1);} marks represent the projective optimistic value. The true values are much closer to the pessimistic bound since the upper bound in \eqref{eq:together} is exact.

\Cref{fig:q01vals} represents the loss of Alice (measured by the true cost as in \eqref{prog:lin_almost_Bayes}) when she plays optimally, pessimistically and optimistically. 
In both figures, the thresholds $\epsilon^-\leq \epsilon^*\leq \epsilon^+$ are represented by gridlines. They correspond to the size of the robust hypothesis at which the cost of full- and no-information, according to each respective program, are equal. This fact is marked by a \tikz[baseline=-0.5ex]{\draw[opacity=1] (0,0) node[cross=3pt, opacity=1] {};} mark at each of the three crossings.

\subsection{The opening example}

We examine the opening example, specifically with parameter $k>\nicefrac12$ and $k\neq1$, through the same lens as the unidimensional case. For this instance, \eqref{eq:together} becomes
\begin{align}
	(1-\beta^2)\epsilon^2 + 2\bar\beta\epsilon|1-k|\mathbb E[\|\hat x\|]
	&\leq \mathbb E\!\left[\max_{\eta\in C\mathcal B} ~w(\eta, \bar\mu)\right] \\
	&\leq \epsilon^2 + 2\epsilon |1-k|\mathbb E[\|\hat x\|],
\end{align}
to be compared with the exact value
\begin{equation}
	\epsilon^2 + 2\epsilon|1-k|\mathbb E[\|\hat x\|].
\end{equation}
Once again, the first pessimistic approximation is exact. In addition, while proving \Cref{thm:boundproj} (via \Cref{lem:prettybound2}), we have obtained the following approximation
\begin{equation}
	\mathbb E[\|\hat x\|]
	\geq \mathbb E[|x_1|] \sqrt{\mathbb E[\|\hat x\|^2]} = \mathbb E[|x_1|] \sqrt{\Tr\Sigma}.
\end{equation}
This bound is slightly tighter than the one used to derive the Projective Optimistic Program thanks to the fact that $l_2=0$ in this specific instance. These two arguments strengthen our expectation that the Pessimistic Program is more accurate than the Projective Optimistic Program.

The Pessimistic Program is strictly concave in $\Tr\Sigma$, hence the solution is either $\Sigma=0$ or $\Sigma=I_n$, thus it suffices to consider these two policies. In the no-information scenario, Jensen's inequality is an equality and so once more, the pessimistic value of the no-information policy is exact. In the full-information scenario, the approximation is not exact, however
\begin{equation}
	1 
	\geq \frac{\mathbb E[\|x\|]}{\sqrt{\mathbb E[\|x\|^2]}} 
	= \underset{\to_n1}{\underbrace{\frac{\sqrt2 \Gamma(\nicefrac{n+1}2)}{\sqrt n \Gamma(\nicefrac n2)}}}
	\geq \sqrt{\frac2\pi}.
\end{equation}

\begin{table}
\centering
\caption{Value of the penalty term according to each program.}
\label{tab:objs2}

\begin{tabular}{|c|c|c|}
\hline
& NI & FI \\ \hline
\eqref{prog:lin_almost_Bayes} & $\epsilon^2$ & $\epsilon^2+2\sqrt2\epsilon |1-k| \frac{\Gamma(\nicefrac{n+1}2)}{\Gamma(\nicefrac n2)}$\\
\eqref{prog:pessimistic} & $\epsilon^2$ & $\epsilon^2 + 2\epsilon |1-k| \sqrt n$ \\
\eqref{prog:optimistic} & $(1-\bar\beta^2)\epsilon^2$ & $(1-\bar\beta^2)\epsilon^2 + 2\bar\beta\kappa\epsilon |1-k| \sqrt n$ \\\hline
\end{tabular}
\end{table}

\Cref{tab:objs2} contains the value of the penalty term of each program for both policies. Once again, in each case, full-information is optimal until a certain threshold is met. The optimal threshold is
\begin{equation}
	\epsilon^* = \frac{(2k-1)n}{2\sqrt2 |1-k| \frac{\Gamma(\nicefrac{n+1}2)}{\Gamma(\nicefrac n2)}},
\end{equation}
whereas the pessimistic and optimistic thresholds are
\begin{equation}
	\epsilon^- = \underset{\to_n 1}{\underbrace{\frac{\sqrt2 \Gamma(\nicefrac{n+1}2)}{\sqrt n \Gamma(\nicefrac n2)}}} \epsilon^*,
	~\epsilon^+ = \underset{\approx3.57}{\underbrace{\frac1{\bar\beta\kappa}}} \epsilon^- = \frac{\sqrt2 \Gamma(\nicefrac{n+1}2)}{\bar\beta\kappa\sqrt n \Gamma(\nicefrac n2)} \epsilon^*.
\end{equation}
The conclusion we drew for the unidimensional setting also applies to the opening example: \eqref{prog:pessimistic} is qualitatively better suited to represent \eqref{prog:lin_almost_Bayes}.

\subsection{A numerical example}

To illustrate the numerical procedure described in \Cref{sec:numerical}, we consider a case where $n=3$, there is no linear term or constant term, and
\begin{equation}
	Q
	= \begin{bmatrix}
		31&-33&51&-5&2&-3\\
		-33&67&-80&4&-9&6\\
		51&-80&112&-7&8&-11\\
		-5&4&-7&1&0&0\\
		2&-9&8&0&2&0\\
		-3&6&-11&0&0&4
	\end{bmatrix} \succ0.
\end{equation}
In this case,
\begin{equation}
	D = \begin{bmatrix}
		-9&6&-10 \\
		6&-16&14 \\
		-10&14&-18
	\end{bmatrix} \prec0,
\end{equation}
The parameters are indeed chosen so that Alice reveals the information fully when $\epsilon=0$, though they are rather arbitrary beyond that. To keep things simple, consider the ellipsoidal hypothesis of parameter $C = \epsilon I_3$. Then, leaving $\epsilon$ out as a factor, $\bar\lambda = 4$, $f=0$ and
\begin{equation}
	E = \begin{bmatrix}
		116& -192& 260\\ 
		-192& 404& -504\\ 
		260& -504& 648
	\end{bmatrix} \succ0.
\end{equation}
The Pessimistic Program is
\begin{equation}
	\Tr Q_{11}+\epsilon^2 \bar\lambda + \min_{0\preceq X\preceq I_3} ~\Tr(DX) + \epsilon\sqrt{\Tr(E\Sigma)}.
\end{equation}

Following the procedure laid out in \Cref{prop:line,prop:suboptim}, we compute the solution at varying $\epsilon$. In \Cref{fig:rank}, we plot the rank of the optimal solution of the Pessimistic Program. Just as shown in \Cref{thm:monotonicity}, the rank never increases with $\epsilon$. At small $\epsilon$ the rank of the solution remains equal to that of the Bayesian solution, whereas at large $\epsilon$ the rank is null as $E\succ0$. Precisely, \Cref{prop:noinf_lim} predicts that whenever
\begin{equation}
	\epsilon \geq \frac{(\sqrt f - \Tr(P_D^{<0} D))^2-f}{\underline\lambda(E)} \approx 6.72,
\end{equation}
$\Sigma=0$ is a solution of the Pessimistic Program. As can be seen on \Cref{fig:rank}, this actually occurs as soon as $\epsilon\geq1.7$.

In \Cref{fig:values}, we plot the values given by the three different programs \eqref{prog:optimistic}, \eqref{prog:genoptimistic}, \eqref{prog:pessimistic}, along with twice the value of \eqref{prog:genoptimistic} (referred to as (2UOP)). We also include the value of the Strong Projective Optimal Program \eqref{prog:spop}, which is a refinement of \eqref{prog:optimistic} using the bound of \Cref{thm:boundproj} with the largest $\beta\in[0,1]$. Technically speaking, its value is
\begin{equation}
	\label{prog:spop}
	\tag{SPOP}
	\min_{0\preceq\Sigma\preceq I_n}~ \max_{\beta\in[0,1]}~ \Tr(D\Sigma) + c+(1-\beta^2)\bar\lambda + \beta\kappa \sqrt{f+\Tr(E\Sigma)}
\end{equation}
which is computed by first resolving the inner maximization, then using similar techniques to those that allow us to solve \eqref{prog:pessimistic} numerically, details are in \Cref{sec:numerical_app}. The dashed red lines denote pessimistic bounds of the optimal cost, the solid blue ones denote optimistic bounds. The \tikz[baseline=-0.5ex]\draw (0,0) circle (2pt); marks highlight the tightest known bounds on the optimal cost restricting to projective policies, whereas the \tikz[baseline=-0.5ex]{\draw[opacity=1] (0,0) node[cross=3pt, opacity=1] {};} marks represent the ones without restriction to projective policies.

For the purposes of bounding the objective of \eqref{prog:Bayes_linquad}, only the tightest robust bound matters. However, the other bounds may provide additional instance-specific guarantees, which might prove---as e.g. in this example---much more informative than the worst-case guarantees laid in \Cref{thm:boundany,thm:boundproj}. Specifically, although the pessimistic bound could be as high as twice the true optimal cost (or at least twice that of \eqref{prog:genoptimistic}), there is a substantial gain in using \eqref{prog:pessimistic} as a robust proxy program for \eqref{prog:lin_almost_Bayes}, as evidenced by the gap between \eqref{prog:pessimistic} and (2UOP). Moreover, \eqref{prog:spop} (and \eqref{prog:genoptimistic}, to a lesser degree) remains quite close to \eqref{prog:pessimistic} so that the true optimal cost, restricting or not to projective policies, remains well controlled---the shaded area represents the uncertainty about the exact value of \eqref{prog:lin_almost_Bayes}.

\begin{figure}
\centering
\begin{tikzpicture}
\begin{axis}[clip=false,
					enlargelimits=false,
					axis lines=left,
		            xlabel=$\epsilon$,
					xlabel style={
					  at={(ticklabel* cs:1)},
				      anchor=north west,
				      below
				    },
					xmin=0,
					xmax=2.5,
					xtick={0,0.5,...,2},
					ytick={0,1,2,3},
		            ymin=0,
		            ymax=3.2,
		            ymajorgrids=true,
		            width=\columnwidth,
		            height=0.5\columnwidth]

	\pgfplotstableread{python/ranks.txt}\ranks;
	\addplot[thick, no marks, blue] table {\ranks};

\end{axis}
\end{tikzpicture}
\caption{Plot of the rank of the solution to the pessimistic program.}
\label{fig:rank}
\end{figure}
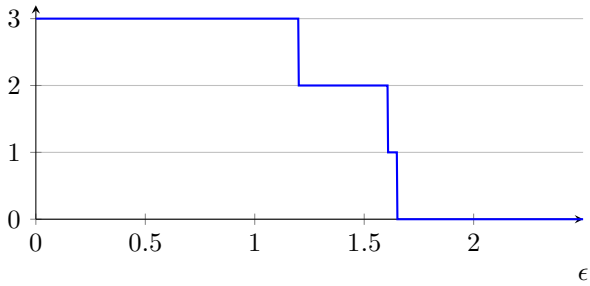

\begin{figure}
\centering
\begin{tikzpicture}
\begin{axis}[clip=false,
					enlargelimits=false,
					axis lines=left,
		            xlabel=$\epsilon$,
					xlabel style={
					  at={(ticklabel* cs:1)},
				      anchor=north west,
				      below
				    },
					xmin=0,
					xmax=2.5,
					xtick={0,0.5,...,2},
					ytick={0,50,...,350},
		            ymin=0,
		            ymax=400,
		            ymajorgrids=true,
		            legend pos=south east,
					legend style={font=\footnotesize},
		            width=\columnwidth]

	\addplot[no marks, red, dashed] file {./python/upess.txt};
	\addlegendentry{(2UOP)}
	\addplot[name path=pess, mark=otimes, mark repeat=200, mark phase = 50, mark options={solid}, red, dashed] file {./python/pess.txt};
	\addlegendentry{\eqref{prog:pessimistic}}
	\addplot[mark=o, mark repeat=200, mark phase = 116, blue] file {./python/spopt.txt};
	\addlegendentry{\eqref{prog:spop}}
	\addplot[no marks, blue] file {./python/popt.txt};
	\addlegendentry{\eqref{prog:optimistic}}
	\addplot[name path=uopt, mark=x, mark repeat=200, mark phase=183, blue] file {./python/uopt.txt};
	\addlegendentry{\eqref{prog:genoptimistic}}
	\addplot[orange!30, opacity=0.5] fill between [of = pess and uopt, split];
	\addlegendentry{\eqref{prog:lin_almost_Bayes}}
\end{axis}
\end{tikzpicture}
\caption{Plot of the various bounds.}
\label{fig:values}
\end{figure}
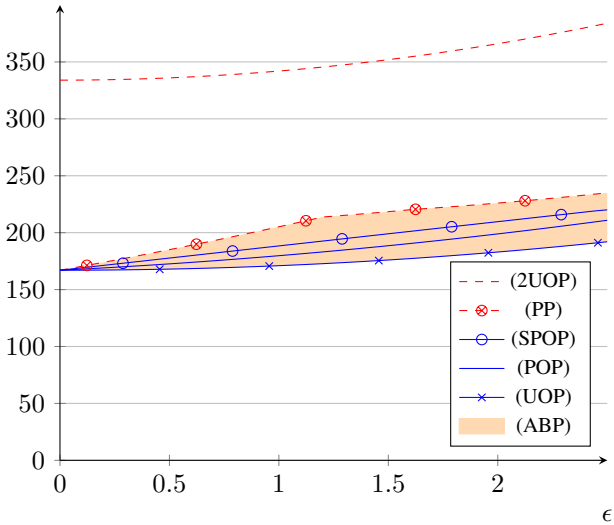

\section{Conclusion}
\label{sec:conclusion}

We have developed and explored the concept of almost-Bayesian agent in a specific persuasion setting: quadratic persuasion. In contrast with previous work, our approach does not assume that the thought process of the Receiver is given and known, but instead that his actions are relatively close to those of a Bayesian agent. This robust concept allows the Sender to account for possible small mistakes the Receiver could commit, either for his inaccuracy in estimating probabilities, or for his failure to exactly optimize his expected utility. Such description of an agent is independent of the form of the event space, the prior or the utilities, and as such is readily transposable to other Bayesian persuasion problems, even though the analysis could greatly differ. 

Even the simplest case of almost-Bayesian quadratic persuasion, exposed in \Cref{sec:cex}, proved to be exactly intractable. Indeed, linear policies---the only practical class of policies for rotationally-invariant priors---have been shown to not be optimal, moreover finding the optimal linear policy is more than challenging. Nonetheless, we could approximate Alice's program (thanks to \Cref{thm:boundany} and \ref{thm:boundproj}) and solve it numerically. In addition, we have uncovered some structural properties of the program, allowed by the specific setting we have chosen. Alice is less keen to share information as Bob's thought process is increasingly departing from Bayesian updating, both truly (\Cref{thm:noinf}) and in approximation (\Cref{thm:monotonicity}). In this case then, failing to be rigorously Bayesian can be detrimental to Bob.

Some of the insights gained in this article are specific to the instance on which we chose to demonstrate the almost-Bayesian agent concept, and partly also to the approximations we derived. In the absence of additional structure however, we suspect that Alice's strategy facing an increasingly less Bayesian would not change consistently. This is similar in spirit to the findings of \cite{de2022non}: over all Bayesian persuasion problems, Alice does not consistently prefer a type of agent, yet, considering more defined instances such as situations with common interest, comparisons can be drawn.

\section*{References}
\bibliography{refs}

\clearpage
\appendices
\crefalias{section}{appendix}
\crefalias{subsection}{appendix}

\section{Origins of the hypothesis class}
\label{sec:hypotheses}
\subsection{Examples of robust hypotheses}

A first natural idea is for Alice to assume that Bob's erroneous posterior lies within a given distance from the Bayesian posterior, as measured by some statistical metric. In the case of the Wasserstein distance, we can state the following result.

\begin{prop}
\label{prop:Wasserstein}
	Let $W_p(\mu',\mu)$ denote the usual Wasserstein distance of order $p\in[1,\infty)$ (see \cite{villani} for a formal definition) between $\mu'$ and $\mu$, and let the robust hypothesis $\Lambda$ be given by
	\begin{equation} \label{eq:lambdawasserstein}
		\Lambda(\mu) = \{\mu', ~W_p(\mu',\mu)\leq\epsilon\},
	\end{equation}
	then
	\begin{equation}
		\bar\Lambda(\mu) = \bar\mu + \epsilon\mathcal B.
	\end{equation}
\end{prop}

In words, $\Lambda$ as in \eqref{eq:lambdawasserstein} corresponds to the robust hypothesis that ``Bob's posterior is always within $W_p$-distance $\epsilon$ from the true posterior.'' It is remarkable that, in terms of means, it induces a simple Euclidean ball. Moreover, $\bar\Lambda(\mu)$ only depends on the mean $\bar\mu$ of $\mu$.

\medskip
Regarding the choice of statistical distance, one can also consider the broad family of $f$-divergences (see \cite{liese1987convex} for a reference). Let $f\colon (0,\infty)\to\mathbb R$ be convex with $f(1)=0$, and interpret $f(0)$ as the limit of $f(\epsilon)$ as $\epsilon>0$ vanishes. We denote the $f$-divergence of $\mu'$ from $\mu$ by
\begin{equation}
	\df{\mu'}{\mu}
	= \int_{\mathbb R^n} f\circ\frac{\mathrm d\mu'}{\mathrm d\mu} \,\mathrm d\mu,
\end{equation}
whenever $\mu' \ll \mu$. For simple instances, with $f(t)=t \ln t$, one recovers the Kullback-Leibler divergence, and with a little more care, one recovers the R\'enyi divergences. 

As it turns out, $f$-divergences prove to be a little more difficult to work with, and we have to restrict our attention to posteriors $\mu$ that stem from a projective policy. Even then, unlike the case of \Cref{prop:Wasserstein}, the mean set $\bar\Lambda(\mu)$ also depends on the covariance $\Sigma_\mu$. More precisely, we can show the following.

\begin{prop}
\label{prop:fdivergence}
	Let the robust hypothesis $\Lambda$ be given by
	\begin{equation}
		\Lambda(\mu) = \{\mu'\ll\mu, ~\df{\mu'}{\mu}\leq\epsilon\},
	\end{equation}
	and let $\mu$ be a Bayesian posterior obtained by a projective policy $P$, then
	\begin{equation}
		\bar\Lambda(\mu) = \bar\mu + \delta (I_n-P)\,\mathcal B,
	\end{equation}
	where the scalar $\delta$ could be infinite (in which case $\bar\Lambda(\mu) = \mathbb R^n$) and implicitly depends on $f$, $\epsilon$ and $\bar\mu$. The ball $\mathcal B$ could be closed.
\end{prop}

When $\nu$ is Gaussian, $\mu = \mathcal N(\bar\mu, I_n-\Sigma)$ is Gaussian as well. What is remarkable is that once centered, all $\mu$ are the same distribution $\mathcal N(0,I_n-\Sigma)$. In this specific case then, $\delta$ does not depend on $\bar\mu$.

\medskip
Another way in which a set of erroneous posteriors can be generated is if Bob is Bayesian but that his computation costs him. In this event, he may very well trade off accuracy for efficacy, and thus be content with a suboptimal solution. As mentioned earlier, Bob's in-game loss is often considered quadratic in linear-preference persuasion, that is
\begin{equation}
\label{eq:Bobquadloss}
	u(a,x) = \begin{bmatrix} x \\ a \end{bmatrix}^\top R \begin{bmatrix} x \\ a \end{bmatrix} + m^\top \begin{bmatrix} x \\ a \end{bmatrix} + s,
\end{equation}
where $R_{22}\succ0$ and, to suit our construction, $R_{12}$ is assumed non-singular. Under belief $\mu$ and with no computation cost, Bob's best-response is,
\begin{equation}
	a^*(\bar\mu) = -R_{22}^{-1}(\nicefrac{m_2}2 + R_{21} \bar\mu).
\end{equation}
With this notation, we can state the following.

\begin{prop}
\label{prop:costlyupdate}
	When Bob's loss is as in \eqref{eq:Bobquadloss},
	\begin{equation}
		\{a, ~u(a,\mu) \leq u(a^*(\bar\mu),\mu) + \epsilon\}
		= a^*\!\left(\bar\mu + \sqrt\epsilon R_{21}^{-1} \sqrt{R_{22}} \mathcal B\right).
	\end{equation}
\end{prop}

In other words, Bob being satisfied with an $\epsilon$-suboptimal solution corresponds exactly to Bob playing optimally but with posteriors such that the set of means is
\begin{equation}
	\bar\Lambda(\mu) = \bar\mu + \sqrt\epsilon R_{21}^{-1} \sqrt{R_{22}} \mathcal B.
\end{equation}
This robust hypothesis is very similar to that of the ``Wasserstein distance'' case, in the sense that we would only need to rescale the Euclidean metric to match it.

\medskip
Finally, instead of a generic model ``$\mu'$ is close to $\mu$,'' Alice can have an idea about Bob's thought process. For instance, she may know that Bob holds a different prior or that he gives more importance to his prior than a Bayesian agent would. At the same time, she may not know his prior exactly or how conservative his belief update is. This direction was recently suggested by \cite{de2022non} while studying non-Bayesian persuasion, i.e., the case where $\Lambda$ is a univalued map, aptly called belief distortion. We discuss here how the type of robust hypothesis introduced in this article can provide useful over-approximation for these so-called parametric models.

As it turns out, not all belief distortion models are well-adapted to uncountable event spaces. For instance Grether's $\alpha-\beta$ model \cite{grether1980bayes} does not generalize to richer event spaces unless $\alpha=1$, and even then, the formula may terminate on an undetermined form, leaving Bob's posteriors undefined. A mismatched prior, on the other hand, poses no apparent technical trouble provided Bob's prior $\nu'$ has finite second moment, \cite{kazikli2021optimal}.   

At any rate, our approach could be deemed too conservative to adequately treat this type of uncertainty. Alice would rather place the adversarial maximization in front of the expectation, as now the failure of Bob to be Bayesian is ``coherent'' across beliefs. This being said, the merit of our robust hypothesis lies in that we can solve the ultimate program it generates, and one could nonetheless include parameter uncertainty in such hypothesis---albeit conservatively.

\subsubsection{Mismatched prior}

If Bob's prior $\nu'\ll\nu$ is such that
\begin{equation}
	\frac{\mathrm d\nu'}{\mathrm d\nu} \in [s, \nicefrac1s],
\end{equation}
for some $s>0$, we can explicitly write Bob's erroneous belief $D_{\nu'}(\mu)$ as a function of the Bayesian belief $\mu$ through
\begin{equation}
	\mathrm dD_{\nu'}(\mu)(x) = \frac{\frac{\mathrm d\nu'}{\mathrm d\nu}(x)}{\int_{\mathbb R^n} \frac{\mathrm d\nu'}{\mathrm d\nu}\,\mathrm d\mu} \,\mathrm d\mu(x).
\end{equation}
When Alice does not know exactly $\nu'$, this gives rise to a robust hypothesis $\Lambda$. However, merely knowing that $\nu'$ is close to $\nu$ in any statistical sense is not enough. Informally, $\nu'$ could differ ever so slightly from $\nu$ on a narrow band of space, thereby inducing a wildly different estimation $\bar\mu'$ from $\bar\mu$ when the message specifies $x$ is in this band. In this case, thus, we require a stronger, more uniform, notion of proximity. When $\nu'\ll\nu$, we let $\epsilon(\nu',\nu)$ be the infimum of all $\epsilon>0$ such that
\begin{equation}
	\frac{\mathrm d\nu'}{\mathrm d\nu} \in \left[\frac1{1+\epsilon}, 1+\epsilon\right].
\end{equation}
The smaller $\epsilon(\nu',\nu)$ is, the closer the distributions are. With this notation in hand, we are in a position to state the following proposition.

\begin{prop}
\label{prop:mismatchedprior}
	Let the robust hypothesis $\Lambda$ be given by
	\begin{equation}
		\Lambda(\mu) = \{D_{\nu'}(\mu), ~\nu' \ll \nu, ~\epsilon(\nu',\nu) \leq \epsilon\},
	\end{equation}
	then
	\begin{equation}
		\bar\Lambda(\mu) \subset \bar\mu + \sqrt{2\epsilon+\epsilon^2} \sqrt{\Tr\Sigma_\mu}\,\mathcal B.
	\end{equation}
\end{prop}

\subsubsection{Affine distortion}

Another model---termed affine distortion---accounts for a bias towards a specific ``ideal'' belief, which may or may not be Bob's prior. Formally, the erroneous belief is
\begin{equation}
	\mu' = \chi \mu + (1-\chi) \mu^*,
\end{equation}
where $\chi\in[0,1]$ is a parameter such that $\chi=1$ corresponds to a Bayesian agent, and $\mu^*$ is the ideal belief. This latter can be interpreted as the belief Bob would like to hold from a motivated updating perspective. Again, a robust hypothesis appears as soon as the parameters are not well-known. For instance, $\chi$ belongs to some subinterval $[a,b]\subset[0,1]$, or $\mu^*$ is close to some belief $\mu^*_0$ in some statistical sense. We explore the latter possibility in the following proposition.

\begin{prop}
\label{prop:affinedistortion}
	Let the robust hypothesis $\Lambda$ be given by
	\begin{equation}
		\Lambda(\mu) = \{\chi \mu + (1-\chi) \mu^*, ~\mu^* \ll \mu^*_0, ~W_p(\mu^*,\mu^*_0)\leq\epsilon\},
	\end{equation}
	then
	\begin{equation}
		\bar\Lambda(\mu) 
		= \chi \bar\mu + (1-\chi) \bar\mu_0^* 
		+ (1-\chi)\epsilon\,\mathcal B.
	\end{equation}
\end{prop}

When $\chi$ is allowed to vary as well, $\bar\Lambda$ takes a rounded cylindrical shape which is perhaps not as convenient to fit in an ellipsoid.

\subsection{Proofs for the Wasserstein distance}

To formally set things, consider $p\geq1$, and two Borel probability measures $P,Q$ on $\mathbb R^n$ with finite $p$-th moment. Denote by $\Gamma(P,Q)$ the space of Borel measures on $\mathbb R^n\times\mathbb R^n$ with marginals $P,Q$ respectively. In this article, we denote by $\|.\|$ the standard Euclidean norm on $\mathbb R^n$. The $p$-Wasserstein distance between $P$ and $Q$ is defined as
\begin{equation}
	W_p(P,Q) = \inf_{\pi\in\Gamma(P,Q)} ~\left(\int_{\mathbb R^n\times\mathbb R^n} \|x-y\|^p\mathrm d\pi(x,y)\right)^{\frac1p}.
\end{equation}

These statistical distances find their origin in optimal transport: the quantity $W_p(P,Q)^p$ corresponds to the minimal cost of displacing a pile of sand distributed as $P$ into another pile distributed as $Q$, where displacing a mass from $x$ to $y$ costs $\|x-y\|^p$. In order to prove \Cref{prop:Wasserstein}, we resort to the following intuitive lemma.

\begin{lemma} \label{lem:wassersteinmean}
	Denoting the mean of $P,Q$ by $\bar P,\bar Q$,	
	\begin{equation}
		W_p(P,Q) \geq \|\bar P-\bar Q\|,
	\end{equation}
	with equality if $Q$ is a translation of $P$.
\end{lemma}

\begin{proof}[Proof of \Cref{lem:wassersteinmean}]
	Let $\pi\in\Gamma(P,Q)$. As the map $(x,y)\mapsto \|x-y\|^p$ is convex, Jensen's inequality yields
	\begin{align}
		&\int_{\mathbb R^n\times\mathbb R^n} \|x-y\|^p\mathrm d\pi(x,y) \\
		&\quad\geq	\left\| \int_{\mathbb R^n\times\mathbb R^n} x\mathrm d\pi(x,y) - \int_{\mathbb R^n\times\mathbb R^n} y\mathrm d\pi(x,y) \right\|^p \\
		&\quad=\|\bar P-\bar Q\|^p.
	\end{align}
	Therefore, as announced,
	\begin{equation}
		W_p(P,Q) \geq \|\bar P-\bar Q\|.
	\end{equation}
	
	If $\mathrm dQ(y) = \mathrm dP(y+x_0)$, we may consider $\pi$ defined by,
	\begin{equation}
		\mathrm d^2\pi(x,y) = \mathrm dP(x) \mathrm d\delta_{x-x_0}(y).
	\end{equation}
	Of course, fixing $A\subset\mathbb R^n$ measurable,
	\begin{align}
		\pi(A\times\mathbb R^n)
		&= \int_A \int_{\mathbb R^n} \mathrm d\delta_{x-x_0}(y)\mathrm dP(x)
		= \int_A \mathrm dP(x) \\
		&= P(A) \\
		\pi(\mathbb R^n\times A)
		&= \int_{\mathbb R^n} \int_A \mathrm d\delta_{x-x_0}(y) \mathrm dP(x) \\
		&= \int_{\mathbb R^n} \mathds 1_A(x-x_0) \mathrm dQ(x-x_0) \\
		&= Q(A),
	\end{align}
	so $\pi\in\Gamma(P,Q)$. On the other hand,
	\begin{align}
		&\int_{\mathbb R^n\times\mathbb R^n} \|x-y\|^p\mathrm d\pi(x,y) \\
		&\quad=\int_{\mathbb R^n}\int_{\mathbb R^n} \|x-y\|^p \mathrm d\delta_{x-x_0}(y) \mathrm dP(x) \\
		&\quad=\int_{\mathbb R^n}\|x_0\|^p \mathrm dP(x) = \|x_0\|^p.
	\end{align}
	As a result,
	\begin{equation}
		W_p(P,Q) \leq \|x_0\| = \|\bar P-\bar Q\|,
	\end{equation}
	so that $W_p(P,Q) = \|\bar P-\bar Q\|$.
\end{proof}

\begin{proof}[Proof of \Cref{prop:Wasserstein}]
	The proof is by double inclusion. Using the first implication of \Cref{lem:wassersteinmean},
\begin{align}
	\bar\Lambda(\bar\mu)
	&= \{\bar\mu', ~W_p(\mu',\mu)\leq\epsilon\} \\
	&\subset \{\bar\mu', ~\|\mu'-\mu\|\leq\epsilon\} \\
	&= \bar\mu + \epsilon\mathcal B.
\end{align}
	On the other hand, let $v = \bar\mu + \epsilon u$ belong to this latter set, i.e., with $u\in\mathcal B$. We may consider the distribution $\mu$ shifted by $\epsilon u$. Surely, by \Cref{lem:wassersteinmean},
	\begin{equation}
		W_p(\mu',\mu) = \|\bar\mu' - \bar\mu \| = \epsilon \|u\| \leq \epsilon,
	\end{equation}
	so $\bar\mu' = \bar\mu + \epsilon u = v \in \bar\Lambda(\bar\mu)$.
\end{proof}

\subsection{Proofs for the $f$-divergences}

R\'enyi divergences can be expressed as a composition of an $f$-divergence by an increasing function. Explicitly, for $\alpha>1$,
\begin{equation}
	\ra{P}{Q}
	= \frac1{\alpha-1} \ln (1+\fdiv{f_\alpha}{P}{Q}),
\end{equation}
with $f_\alpha(t) = t^\alpha-1$ and for $\alpha\in(0,1)$,
\begin{equation}
	\ra{P}{Q}
	= \frac1{1-\alpha} \ln \frac1{1-\fdiv{f_\alpha}{P}{Q}},
\end{equation}
with $f_\alpha(t) = 1-t^\alpha$.

\begin{lemma}
\label{lem:changeofvariable}
	Let $\phi\colon \mathcal X \to \mathcal Y$ be a measurable injection, $P,Q$ be probability measures on $\mathcal X$ such that $P\ll Q$, and $f$ convex with $f(1) = 0$. Then, the $f$-divergence of the pushforward of $P$ by $\phi$ from the pushforward of $Q$ by $\phi$ is the $f$-divergence of $P$ from $Q$:
\begin{equation}
	\df{\phi_*P}{\phi_*Q} = \df{P}{Q}.
\end{equation}
\end{lemma} 

\begin{proof}[Proof of \Cref{lem:changeofvariable}]
	It is a straightforward change of variable. We first verify that $Q$-almost everywhere
	\begin{equation}
		\frac{\mathrm d(\phi_*P)}{\mathrm d(\phi_*Q)} \circ \phi = \frac{\mathrm dP}{\mathrm dQ}.
	\end{equation}
	Let $A\subset \mathcal X$ be measurable, by injectivity $\phi^{-1}(\phi(A)) = A$ and so
	\begin{align}
		\int_A \frac{\mathrm d(\phi_*P)}{\mathrm d(\phi_*Q)} \circ \phi \,\mathrm dQ
		&= \int_{\phi(A)} \frac{\mathrm d(\phi_*P)}{\mathrm d(\phi_*Q)} \,\mathrm d(\phi_*Q) \\
		&= \phi_*P(\phi(A)) \\
		&= P(A).
	\end{align}
	Using this fact,
	\begin{align}
		\df{\phi_*P}{\phi_*Q}
		&= \int_{\mathcal Y} f\circ\frac{\mathrm d(\phi_*P)}{\mathrm d(\phi_*Q)} \,\mathrm d(\phi_*Q) \\
		&= \int_{\mathcal X} f\circ\frac{\mathrm d(\phi_*P)}{\mathrm d(\phi_*Q)} \circ\phi \,\mathrm dQ \\
		&= \int_{\mathcal X} f\circ\frac{\mathrm dP}{\mathrm dQ} \,\mathrm dQ \\
		&= \df{P}{Q}.\qedhere
	\end{align}
\end{proof}

\begin{proof}[Proof of \Cref{prop:fdivergence}]
	Let $\mu$ be a projective belief, it is the result of message $y = \Sigma x$. As a result, $\mu$ is a distribution with support in $y + \ker \Sigma$. In particular, whenever $\mu' \ll \mu$, its support also lies in $y + \ker \Sigma$ and by convexity, $\bar\mu' \in y + \ker \Sigma$. Since $\mu\ll\mu$, $\bar\mu \in y + \ker\Sigma$ as well, so we conclude that
	\begin{equation}
		\bar\Lambda(\bar\mu) \subset \bar\mu + \ker\Sigma.
	\end{equation}
	
	Now, consider a rotation $O \in O(\ker \Sigma)$ that leaves $(\ker \Sigma)^\perp = \Ima \Sigma$ invariant. Proceed to a rotation of the space so that ``$x^* = Ox$ is the new $x$.'' The belief $O_*\mu$ is then the belief obtained when the prior is $O_*\nu = \nu$ and the message is $y=\Sigma x = \Sigma x^*$, in other words, it is $\mu$ itself: $O_*\mu=\mu$. In particular $\bar\mu$ is left invariant by all $O \in O(\ker \Sigma)$, thus $\bar\mu\in\Ima\Sigma$ and so $\bar\mu=y$.
	
	We are now in a position to show that $\bar\Lambda(\bar\mu)$ is invariant by $O(\ker \Sigma)$. Let $O\in O(\ker\Sigma)$ and $m \in \bar\Lambda(\bar\mu)$. This latter is the mean of some $\mu'\in\Lambda(\mu)$. In turn $O_*\mu' \ll O_*\mu=\mu$ also satisfies
	\begin{equation}
		\df{O_*\mu'}{\mu} = \df{O_*\mu'}{O_*\mu} = \df{\mu'}{\mu} \leq \epsilon,
	\end{equation}
	that is $O_*\mu' \in \Lambda(\mu)$ and so $O\bar\mu = Om \in \bar\Lambda(\bar\mu)$.
	
	All in all, this shows that
	\begin{equation}
		\bar\Lambda(\bar\mu) = \bar\mu + \delta (I_n-\Sigma) \mathcal B,
	\end{equation}
	where $\delta\geq0$ could be ``infinite'' and the ball $\mathcal B$ could actually be closed. This point matters less to Alice since the objective $w(.,\bar\mu)$ is continuous. 
\end{proof}

\subsection{Proofs for the costly update}

\begin{proof}[Proof of \Cref{prop:costlyupdate}]
	First rewrite the cost by completing the square,
	\begin{equation}
		u(a,\mu)
		= (a-a^*(\bar\mu))^\top R_{22} (a-a^*(\bar\mu)) + o,
	\end{equation}
	where $o$ is a constant. As a result,
	\begin{align}
		\{a, ~u(a,\mu) \leq u(a^*(\bar\mu),\mu) + \epsilon\}
		&= a^*(\bar\mu) + \sqrt\epsilon \sqrt{R_{22}}^{-1} \mathcal B \\
		&= a^*\!\left(\bar\mu + \sqrt\epsilon R_{21}^{-1} \sqrt{R_{22}} \mathcal B\right).
	\end{align}
\end{proof}

\subsection{Proof for the parametric models}
\label{sec:appendixparametric}

\begin{proof}[Proof of \Cref{prop:mismatchedprior}]
	We first explain the formula we had announced. Bayes' rule is better characterized in terms of joint probabilities. The distribution $\tau$ of posteriors $\mu_y$ is the essentially unique one such that
	\begin{equation}
		\mathrm d\sigma_x(y) \mathrm d\nu(x) = \mathrm d\mu_y(x) \mathrm d\tau(y).
	\end{equation}
	A nitty-gritty discussion would dive into the technical details of this definition, where notably the disintegration theorem would be of great help (see \cite{disintegration}), but we choose to remain informal for the proof of this relatively less important proposition. In this context then,
	\begin{equation}
		\frac{\mathrm d\nu'}{\mathrm d\nu}(x) = \frac{\mathrm d\tau'}{\mathrm d\tau}(y) \frac{\mathrm d\mu'_y}{\mathrm d\mu_y}(x).
	\end{equation}
	The formula then follows from the fact that $\mu'_y$ is a probability measure.
	
	Let then $\nu'\ll\nu$ be such that $\epsilon(\nu',\nu)\leq\epsilon$. The mean difference between the distorted belief and the Bayesian belief is
	\begin{equation}
		\overline{\mathrm dD_{\nu'}(\mu)} - \bar\mu = \int_{\mathbb R^n} x \left(\frac{\frac{\mathrm d\nu'}{\mathrm d\nu}(x)}{\int_{\mathbb R^n} \frac{\mathrm d\nu'}{\mathrm d\nu}\,\mathrm d\mu}-1\right) \,\mathrm d\mu(x).
	\end{equation}
	The condition $\epsilon(\nu',\nu)\leq\epsilon$ implies that the bracketed term has magnitude at most $\sqrt{2\epsilon+\epsilon^2}$. The Cauchy-Schwarz inequality then yields
	\begin{equation}
		\left\| \overline{\mathrm dD_{\nu'}(\mu)} - \bar\mu \right\|
		\leq \sqrt{2\epsilon+\epsilon^2} \sqrt{\Tr \Sigma_\mu},
	\end{equation}
	which establishes the inclusion.
\end{proof}

\begin{proof}[Proof of \Cref{prop:affinedistortion}]
	Observe that
	\begin{equation}
		\Lambda(\mu) = \chi \mu + (1-\chi) \{ \mu^*, ~\mu^* \ll \mu^*_0, ~W_p(\mu^*,\mu^*_0)\leq\epsilon\},
	\end{equation}
	the last set is none other than $\Lambda$ in the case of Wasserstein distance. In turn,
	\begin{equation}
		\bar\Lambda(\mu) 
		= \chi \bar\mu + (1-\chi) \bar\mu_0^* 
		+ (1-\chi)\epsilon\,\mathcal B,
	\end{equation}
	as stated.
\end{proof}

\section{The non-Bayesian programs}

\subsection{Rewriting the true program}

\begin{proof}[Proof of \Cref{lem:nonneg}]
	Substituting $a=B\tilde x+b$ yields
	\begin{equation}
		\begin{bmatrix} x \\ a \end{bmatrix}^\top M \begin{bmatrix} x \\ a \end{bmatrix} + p^\top \begin{bmatrix} x \\ a \end{bmatrix} + q
		= \begin{bmatrix} x\\\tilde x \end{bmatrix} Q \begin{bmatrix} x\\\tilde x \end{bmatrix} + (p')^\top \begin{bmatrix} x\\\tilde x \end{bmatrix} + q',
	\end{equation}
	where
	\begin{align}
		Q &= \begin{bmatrix} M_{11} & M_{12} B \\ B^\top M_{21} & B^\top M_{22} B \end{bmatrix} \\
		p' &= \begin{bmatrix} p_1 + 2M_{12}b \\ B^\top p_2 + 2B^\top M_{22}b \end{bmatrix} \\
		q' &= q + b^\top M_{22} b + p_2^\top b.
	\end{align}
	Let $u$ be a vector of dimension $2n$. Considering the above nonnegative form at $(x,\tilde x) = t u$ with any $t$ yields that $u^\top Q u \geq 0$. In turn, $Q\succeq0$. Moreover, if $Q u = 0$ for some vector $u$, then considering the above form with again $(x,\tilde x) = t u$ yields that $(p')^\top u =0$ as well. In other words,
	\begin{equation}
		\ker Q \subset (p')^\perp,
	\end{equation}
	thus
	\begin{equation}
		p' \in (\ker Q)^\perp = \Ima Q^\top.
	\end{equation}
	There thus exists some vector $l$ so that
	\begin{equation}
		p' = -2Q^\top l,
	\end{equation}
	and so,
	\begin{equation}
		\begin{bmatrix} x \\ a \end{bmatrix}^\top M \begin{bmatrix} x \\ a \end{bmatrix} + p^\top \begin{bmatrix} x \\ a \end{bmatrix} + q
		= \left(\begin{bmatrix} x\\\tilde x \end{bmatrix} - l\right)^\top Q \left(\begin{bmatrix} x\\\tilde x \end{bmatrix} - l \right) + r,
	\end{equation}
	having let
	\begin{equation}
		r = q' - l^\top Q l.
	\end{equation}
	Finally, considering the above nonnegative form with $(x,\tilde x)=l$ yields $r\geq0$.
\end{proof}

Akin to \Cref{lem:rewriteTamura}, Alice first rewrites the objective of her program in the non-Bayesian case, this is the object of \Cref{lem:rewrite}. The proof uses the reductions $\bar\nu = 0$ and $\Sigma_\nu=I_n$.

\begin{proof}[Proof of \Cref{lem:rewrite}]
	Begin by rewriting the objective of \eqref{prog:hatv_lin} being maximized,
	\begin{align}
		v(a(\bar\mu'),\mu)
		=&~ \mathbb E_\mu\!\left[\left(\begin{bmatrix}x\\\bar\mu'\end{bmatrix}-l\right)^\top Q \left(\begin{bmatrix}x\\\bar\mu'\end{bmatrix}-l\right)\right] + r \\
		=&~ \eta^\top Q_{22} \eta
		+ 2\begin{bmatrix}0\\\eta\end{bmatrix}^\top Q  \left(\begin{bmatrix}\bar\mu\\\bar\mu\end{bmatrix}-l\right) \\
		&+ \mathbb E_\mu\!\left[\left(\begin{bmatrix}x\\\bar\mu\end{bmatrix}-l\right)^\top Q \left(\begin{bmatrix}x\\\bar\mu\end{bmatrix}-l\right)\right] + r.
	\end{align}
	Clearly, this depends quadratically on $\eta = \bar\mu' - \bar\mu$. The quadratic coefficient is constant, and the linear coefficient solely depend on $\bar\mu$. If we average the coefficient that is constant with respect to $\eta$, we obtain the Bayesian objective
	\begin{align}
		&\mathbb E\!\left[\left(\begin{bmatrix}x\\\bar\mu\end{bmatrix}-l\right)^\top Q \left(\begin{bmatrix}x\\\bar\mu\end{bmatrix}-l\right)\right] + r \\
		&\quad= \mathbb E\!\left[\begin{bmatrix}x\\\bar\mu\end{bmatrix}^\top Q \begin{bmatrix}x\\\bar\mu\end{bmatrix}\right] 
		+ l^\top Q l
		+ r \\
		&\quad= \Tr(D\Sigma) + \Tr Q_{11} + l^\top Q l + r.
	\end{align}
	where again $\Sigma = \mathbb E_{\bar\tau}[(\bar\mu-\bar\nu)(\bar\mu-\bar\nu)^\top]$ is the covariance of the estimate, as before. On the other hand, we may develop the linear and quadratic term in $\eta$,
	\begin{equation}
			w(\eta, \bar\mu)
			= 2((Q_{21}+Q_{22})\bar\mu -Q_{21}l_1-Q_{22}l_2)^\top \eta +  \eta^\top Q_{22} \eta,
	\end{equation}
	as stated.
\end{proof}

\subsection{The no-information theorems}

\begin{proof}[Proof of \Cref{thm:noinf}]
	Following \Cref{lem:trsols}, $\Sigma=0$ is a solution of \eqref{prog:Bayes_linquad} if and only if $P_D^{<0} = 0$, that is if and only if $D\succeq0$. In this case, we like to rewrite the objective of \eqref{prog:lin_almost_Bayes} as
	\begin{equation}
		\mathbb{E}_{\bar\tau}\!\left[\bar\mu^\top D \bar\mu + c+\max_{\eta\in C\mathcal B} ~w(\eta, \bar\mu)\right].
	\end{equation}
	All the terms inside the expectation are convex in $\bar\mu$, this rather clear for the two first ones. Regarding the last term, let $\bar\mu_1,\bar\mu_2 \in \mathbb R^n$ and $\lambda\in[0,1]$, we have
	\begin{align}
		&\max_{\eta\in C\mathcal B} ~w(\eta, \lambda\bar\mu_1+(1-\lambda)\bar\mu_2) \\
		&\quad=\max_{\eta\in C\mathcal B} ~\lambda w(\eta, \bar\mu_1) + (1-\lambda) w(\eta, \bar\mu_2) \\
		&\quad\leq\lambda \max_{\eta\in C\mathcal B} ~w(\eta, \bar\mu_1) + (1-\lambda) \max_{\eta\in C\mathcal B} ~w(\eta, \bar\mu_2).
	\end{align}
	Convexity being established, we may use Jensen's inequality,
	\begin{align}
		&\mathbb{E}_{\bar\tau}\!\left[\bar\mu^\top D \bar\mu + c+\max_{\eta\in C\mathcal B} ~w(\eta, \bar\mu)\right] \\
		&\quad\geq\mathbb{E}_{\delta_{\bar\nu}}\!\left[\bar\mu^\top D \bar\mu + c+\max_{\eta\in C\mathcal B} ~w(\eta, \bar\mu)\right].
	\end{align}
	The distribution $\delta_{\bar\nu}$ informally corresponds to substituting $\bar\mu$ with its average, $\bar\nu$. This distribution is the result of the no-information policy, for which the estimate is constantly $\bar\nu$.
\end{proof}

\begin{proof}[Proof of \Cref{thm:inf}]
Consider the nested program of the error term of \eqref{prog:lin_almost_Bayes}. Surely
\begin{equation}
	\max_{\eta\in C\mathcal B} ~w(\eta,\bar\mu) 
	= \max_{\eta\in \mathcal B} ~2v^\top \eta +  \eta^\top C^\top Q_{22} C \eta
\end{equation}
where
\begin{equation}
	v = C^\top((Q_{21}+Q_{22})\bar\mu -Q_{21}l_1-Q_{22}l_2).
\end{equation}
The largest eigenvalue of $C^\top Q_{22} C$ is $\bar\lambda$, let $P$ be the orthogonal projection on the corresponding eigenspace. If $Pv\neq0$, consider the argument $\eta = \nicefrac{Pv}{\|Pv\|}$, it yields
\begin{equation}
	\max_{\eta\in C\mathcal B} ~w(\eta,\bar\mu) 
	\geq \bar\lambda + 2\|Pv\|.
\end{equation}
If $Pv=0$, considering any $\eta$ of unit length in the principal eigenspace (i.e., such that $P\eta=\eta$) as an argument yields the same lower bound.

For a converse bound, we first resort to \Cref{lem:sproc}, defined and proved soon below. With the help of an S-procedure (see \cite{boyd1994linear} for a survey), it shows that
\begin{equation}
	\max_{\eta\in C\mathcal B} ~w(\eta,\bar\mu) 
	= \inf_{\lambda>\bar\lambda} ~\lambda + v^\top (\lambda I_n - C^\top Q_{22} C)^{-1} v.
\end{equation}
Considering the argument $\bar\lambda + \|Pv\|$ when $Pv\neq0$ yields
\begin{align}
	&\max_{\eta\in C\mathcal B} ~w(\eta,\bar\mu) \\
	&\quad\leq \bar\lambda + \|Pv\| + v^\top (\bar\lambda I_n - C^\top Q_{22} C + \|Pv\| I_n)^{-1} v \\
	&\quad= \bar\lambda + \|Pv\| + (Pv)^\top (\bar\lambda I_n - C^\top Q_{22} C + \|Pv\| I_n)^{-1} Pv \\
	&\quad\phantom{=}+ (v-Pv)^\top (\bar\lambda I_n - C^\top Q_{22} C + \|Pv\| I_n)^{-1} (v-Pv) \\
	&\quad\leq \bar\lambda + 2\|Pv\| + \frac{\|(I_n-P)v\|^2}{\bar\lambda-\bar\lambda_2}.
\end{align}
When $Pv=0$, for $\lambda>\bar\lambda$,
\begin{equation}
	\lambda + v^\top (\lambda I_n - C^\top Q_{22} C)^{-1} v
	\leq \lambda + \frac{\|v-Pv\|^2}{\lambda-\bar\lambda_2},
\end{equation}
and therefore, letting $\lambda$ tend to $\bar\lambda$, we obtain the same bound as before. 

As $4\mathbb E[\|v\|^2]=f+\Tr E$, taking the expectation of both bounds yields
\begin{align}
	\bar\lambda + 2\mathbb E[\|Pv\|] + \frac{f + \Tr E}{4(\bar\lambda-\bar\lambda_2)}
	&\geq \mathbb E_{\bar\tau}\!\left[\max_{\eta\in C\mathcal B} ~w(\eta,\bar\mu) \right]\\
	&\geq \bar\lambda + 2\mathbb E[\|Pv\|].
\end{align}
The no-information policy costs at least
\begin{equation}
	c + \bar\lambda + 2\|P\mathbb E[v]\|.
\end{equation}
On the other hand, there exists $u$ unit-vector such that
\begin{equation}
	PC^\top(Q_{21}+Q_{22})u = 0
\end{equation}
since that matrix is singular. The policy projective policy $y = uu^\top x$ induces the estimate $\bar\mu = (u^\top x) u$ and thus costs at most
\begin{align}
	&c + u^\top D u + \bar\lambda + 2\|P\mathbb E[v]\| + \frac{f + \Tr E}{4(\bar\lambda-\bar\lambda_2)}\\
	&\quad< c + \bar\lambda + 2\|P\mathbb E[v]\|.\qedhere
\end{align}
\end{proof}

\Cref{thm:noinf} highlights situations where Alice does not want to share any information and \Cref{thm:inf} where Alice's best course of action involves some signaling. The respective conditions are mutually exclusive, of course, but we can show that Alice ceases to share information with the pessimistic approximation, in cases where, optimally, she would still transmit some information.

\begin{prop}
\label{prop:noinf_lim}
	Whenever
	\begin{equation}
		E \succeq \left(\left(\sqrt f-\Tr (DP_D^{<0})\right)^2 - f\right) I_n,
	\end{equation}
	$\Sigma=0$ is a solution of \eqref{prog:pessimistic}. 
\end{prop}

We remind the reader that $P_D^{<0}$ denotes the orthogonal projection on the negative eigenspace of $D$, so that $\Tr(P_D^{<0} D) \leq0$. This proposition states that provided $E$ is large enough, not sending information is optimal among projective policies, from a pessimistic point of view. One can interpret this result in the light of the parametrized hypothesis presented in \Cref{thm:monotonicity}, i.e., of shape $\epsilon^2 CC^\top$. When $E\succ0$, the condition of \Cref{prop:noinf_lim} is satisfied for $\epsilon$ large enough since the left-hand side grows with order $\epsilon^2$, whereas the right-hand side grows with order $\epsilon$. As a result, the solution of \eqref{prog:pessimistic} is $\Sigma=0$, when Bob is not Bayesian enough. 

This contrasts with \Cref{thm:inf} whose condition is independent of $\epsilon$, and hence insures that there are cases where Alice benefits from signaling no matter the value of $\epsilon$. This shows a limit of the Pessimistic Program \eqref{prog:pessimistic} when $\epsilon$ is very large.

\begin{proof}[Proof of \Cref{prop:noinf_lim}]
	By concavity, \eqref{prog:pessimistic} admits a solution which is an orthogonal projection matrix. For $X$ such matrix with rank $\rk X \geq 1$,
	\begin{equation}
		\Tr(DX) + c + \sqrt{f+\Tr(EX)} \geq c + \sqrt f.
	\end{equation}
	The latter being the value of \eqref{prog:pessimistic} at $\Sigma=0$, we deduce that $0$ is a solution.
\end{proof}

\subsection{Technical lemmas}

The first technical lemma consist in turning the inner maximization of \eqref{prog:lin_almost_Bayes} into a univariate convex program; this is the object of the following lemma.

\begin{lemma} \label{lem:sproc}
	Given $Q$ a positive semi-definite matrix, $C$ a matrix and $v$ a vector of appropriate dimensions,
	\begin{equation}
		\max_{\eta\in \mathcal B} ~\eta^\top Q \eta + 2v^\top \eta
		= \inf_{\lambda>\overline\lambda(Q)} ~\lambda + v^\top (\lambda I_n - Q)^{-1} v.
	\end{equation}
\end{lemma}

This can be readily applied to our problem with $C^\top Q_{22} C$ instead of $Q$ and
\begin{equation}
	v = C^\top((Q_{21}+Q_{22})\bar\mu -Q_{21}l_1-Q_{22}l_2).
\end{equation} 
After substitution,
\begin{equation}
	\max_{\eta\in C\mathcal B} ~w(\eta, \bar\mu)
	= \inf_{\lambda>\bar\lambda} ~\lambda + v^\top (\lambda I_n - C^\top Q_{22} C)^{-1} v.
\end{equation}
The appeal of this expression is that it is a one-dimensional convex program, thus for given parameters it is inexpensive to compute its value. Of course, this is merely a first step since this value is to be averaged over all $\bar\mu$. Another advantage of this program is that we can actually provide upper and lower bounds matched up to a constant ratio not so far from $1$.

\begin{lemma} \label{lem:prettybound1}
	Given $Q\succeq0$, $v\in\mathbb R^n$, for all $\beta\in[0,1]$ we have
	\begin{align}
		\overline\lambda(Q) + 2\|v\| 
		&\geq \inf_{\lambda>\overline\lambda(Q)} ~\lambda + v^\top (\lambda I_n - Q)^{-1} v \\
		&\geq (1-\beta^2)\overline\lambda(Q) + 2\beta\|v\|.
	\end{align}
\end{lemma}

Of course $\beta$ can be selected carefully so to match the bounds up to a constant, but we will rather set $\beta$ at our convenience later to combine better with further approximations. For Alice, this means that for all $\beta\in[0,1]$,
\begin{equation}
\tag{\ref{eq:together}}
\begin{aligned}
	(1-\beta^2)\bar\lambda + 2\beta\mathbb E[\|v\|] 
	&\leq \mathbb E\!\left[\max_{\eta\in C\mathcal B} ~w(\eta, \bar\mu)\right] \\
	&\leq \bar\lambda + 2\mathbb E[\|v\|],
\end{aligned}
\end{equation}
where
\begin{align}
	\bar\lambda &= \overline\lambda(C^\top Q_{22} C)\\
	v &= C^\top((Q_{21}+Q_{22})\bar\mu -Q_{21}l_1-Q_{22}l_2).
\end{align}

The next step, of course, is to obtain a good estimate of $\mathbb E[\|v\|]$. Jensen's inequality directly yields
\begin{equation}
	\mathbb E[\|v\|] \leq \sqrt{\mathbb E[\|v\|^2]},
\end{equation}
this can readily be used for the Pessimistic Program, since it only depends on $\Sigma$, $v$ being an affine function of $\bar\mu$. On the other hand, $\bar\mu$ could a priori take on any form, so we cannot hope for a good general converse inequality. Nonetheless, we may take $\beta=0$ and obtain a strong enough lower bound. Otherwise, we can restrict our attention to \emph{projective policies,} i.e., those for which $\hat x = Px$, this is the object of the following lemma.

\begin{lemma} \label{lem:prettybound2}
	When $v$ is an affine function of $x$,
	\begin{equation}
		\mathbb E[\|v\|]
	\geq \kappa \sqrt{\mathbb E[\|v\|^2]},
	\end{equation}
	denoting the first coordinate of $x$ by $x_1$, and
	\begin{equation}
		\kappa = \frac{\mathbb E[|x_1|]}{\sqrt{1+\mathbb E[|x_1|]^2}}.
	\end{equation}
\end{lemma}

For the sake of simplicity, we are brought to introduce
\begin{equation}
    \begin{aligned}
    	E &= 4(Q_{12}+Q_{22}) CC^\top(Q_{21}+Q_{22})\\
    	f &= 4(l_1^\top Q_{12} + l_2^\top Q_{22})CC^\top (Q_{21}l_1+Q_{22}l_2).
    \end{aligned}
\end{equation}
With these notations,
\begin{equation}
	4\|\mathbb E[v]\|^2 = f, ~4\mathbb E[\|v\|^2] = f + \Tr(E\Sigma).
\end{equation}
Combining \Cref{lem:sproc,lem:prettybound1} at $\beta=0$ (as in \eqref{eq:together}) and Jensen's inequality yields the two first inequalities of \Cref{thm:boundany}. Using \Cref{lem:sproc,lem:prettybound1,lem:prettybound2} and Jensen's inequality, we obtain the first two inequalities of \Cref{thm:boundproj}. To obtain the last inequalities, we resort to the following lemma at $\beta=0$ for \Cref{thm:boundany} and at $\beta=\bar\beta$ for \Cref{thm:boundproj}.

\begin{lemma}
\label{lem:posbound}
	Given the previous definitions, for all $\beta\in[0,1]$,
	\begin{align}
		&c+\Tr(D\Sigma)+\bar\lambda+\sqrt{f+\Tr(E\Sigma)}\\
		&\quad\leq \gamma(\beta)(c+\Tr(D\Sigma)+(1-\beta^2)\bar\lambda+\beta\kappa\sqrt{f+\Tr(E\Sigma)}),
	\end{align}
	where
	\begin{equation}
		\gamma(\beta) = \begin{cases}
			\frac{2 - \beta^2 - 2 \beta \kappa}{1 - \beta^2 (1 + \kappa^2)} &\text{if } 1 - \beta^2 - \beta \kappa > 0\\
			\frac1{1 - \beta^2} &\text{otherwise}
		\end{cases}
	\end{equation}
	reaches a minimum at $\bar\beta$ with value $\bar\gamma$.
\end{lemma}

\subsection{Proofs of the technical lemmas}

\begin{proof}[Proof of \Cref{lem:sproc}]
First let
\begin{equation}
	F_1 = \begin{bmatrix} -1&0\\0&I_n\end{bmatrix},
	~F_2(t) = \begin{bmatrix}
		-t&v^\top\\
		v&Q
	\end{bmatrix},
\end{equation}
so that $\eta\in \mathcal B$ if and only if
\begin{equation}
	\begin{bmatrix} 1\\\eta\end{bmatrix}^\top
	F_1
	\begin{bmatrix} 1\\\eta\end{bmatrix} \leq 0,
\end{equation}
and moreover
\begin{equation}
	\eta^\top Q\eta + 2v^\top\eta - t = \begin{bmatrix} 1\\\eta\end{bmatrix}^\top F_2(t) \begin{bmatrix} 1\\\eta\end{bmatrix}.
\end{equation}
By the S-lemma (see \cite{boyd1994linear}),
\begin{align}
	&\big(\eta\in \mathcal B \implies \eta^\top Q \eta + 2v^\top \eta - t\leq0\big) \\
	&\quad\iff ~\big(\exists \lambda\geq0, ~\lambda F_1 \succeq F_2(t)\big),
\end{align}
so we can rewrite
\begin{align}
	&\max_{\eta\in \mathcal B} ~\eta^\top Q \eta + 2v^\top \eta \\
	&\quad= \begin{aligned}[t]
		\min_t &\quad t\\
		\mathrm{s.t.} &\quad \eta\in \mathcal B \implies \eta^\top Q \eta + 2v^\top \eta -t\leq0
	\end{aligned} \\
	&\quad= \begin{aligned}[t]
		\min_{\lambda,t} &\quad t. \\
		\mathrm{s.t.} &\quad \lambda\geq0 \\
		&\quad\lambda F_1 \succeq F_2(t)
	\end{aligned}
\end{align}

We notice that
\begin{equation}
	(\lambda+\epsilon) F_1 - F_2(t+2\epsilon) = \lambda F_1 - F_2(t) + \epsilon I_{n+1},
\end{equation}
therefore if $\lambda F_1 \succeq F_2(t)$, then for all $\epsilon>0$,
\begin{equation}
	(\lambda+\epsilon) F_1 \succ F_2(t+2\epsilon).
\end{equation}
Conversely, if the above holds, then at the limit where $\epsilon$ vanishes, $\lambda F_1 \succeq F_2(t)$. We may thus write
\begin{equation}
	\max_{\eta\in \mathcal B} ~\eta^\top Q \eta + 2v^\top \eta
	= \begin{aligned}[t]
		\inf_{\lambda,t} &\quad t. \\
		\mathrm{s.t.} &\quad \lambda>0 \\
		&\quad\lambda F_1 \succ F_2(t).
	\end{aligned}
\end{equation}

By Schur complement (see \cite{boyd1994linear}), $\lambda F_1 \succ F_2(t)$ if and only if
\begin{equation}
	\begin{cases} 
		\lambda I_n - Q \succ 0\\ 
		-\lambda+t - v^\top (\lambda I_n - Q)^{-1} v > 0.
	\end{cases}
\end{equation}
The first condition boils down to $\lambda>\bar\lambda(Q)$, and, as a result,
\begin{equation}
	\max_{\eta\in \mathcal B} ~\eta^\top Q \eta + 2v^\top \eta
	= \inf_{\lambda>\bar\lambda(Q)} ~\lambda +v^\top (\lambda I_n - Q)^{-1} v,
\end{equation}
concluding the proof.
\end{proof}

\begin{proof}[Proof of \Cref{lem:prettybound1}]
	When $v=0$, the upper bound is trivial. When $v\neq0$, we may substitute
	\begin{equation}
		\lambda = \overline\lambda(Q) + \|v\|,
	\end{equation}
	and obtain
	\begin{align}
		&\inf_{\lambda>\overline\lambda(Q)} ~\lambda + v^\top (\lambda I_n -Q)^{-1} v \\
		&\quad\leq\overline\lambda(Q) + \|v\| + v^\top (\|v\| I_n + \overline\lambda(Q) I_n - Q)^{-1} v \\
		&\quad\leq \overline\lambda(Q) + 2\|v\|.
	\end{align}
	
	Conversely, forget $Q,v$ for a second and fix some $\gamma>0$, we have
	\begin{align}
		&\inf_{\substack{Q\succeq0 \\ \bar\lambda = \overline\lambda(Q) \\ v\neq0}}
		~\frac{\inf_{\lambda>\bar\lambda} ~\lambda + v^\top (\lambda I_n - Q)^{-1} v}
		{\bar\lambda + 2\gamma\|v\|} \\
		&\quad=\inf_{\substack{Q\succeq0 \\ \bar\lambda > \overline\lambda(Q) \\ v\neq0}}
		~\frac{\inf_{\lambda>\bar\lambda} ~\lambda + v^\top (\lambda I_n - Q)^{-1} v}
		{\bar\lambda + 2\gamma\|v\|} \\
		&\quad=\inf_{\substack{\lambda>\bar\lambda>0 \\ \bar\lambda I_n\succeq Q\succeq0 \\ v\neq0}}
		~\frac{\lambda + v^\top (\lambda I_n - Q)^{-1} v}
		{\bar\lambda + 2\gamma\|v\|} \\
		&\quad=\inf_{\substack{\lambda>\bar\lambda>0 \\ v\neq0}}
		~\frac{\lambda + \frac{v^\top v}\lambda}
		{\bar\lambda + 2\gamma\|v\|} \\
		&\quad=\inf_{\lambda,r>0}
		~\frac{\lambda + \frac{r^2}\lambda}
		{\lambda + 2\gamma r} \\
		&\quad=\inf_{t>0}
		~\frac{1 + t^2}
		{1 + 2\gamma t} \\
		&\quad=\frac{\sqrt{1+4\gamma^2}-1}{2\gamma^2}.
	\end{align}
	For a more legible result, we let
	\begin{equation}
		\beta = \frac{\sqrt{1+4\gamma^2}-1}{2\gamma} \in (0,1),
	\end{equation}
	so that
	\begin{equation}
		\gamma = \frac{\beta}{1-\beta^2}.
	\end{equation}
	As a result, for all $Q\succeq0$, $v\in\mathbb R^n$ and $\beta\in[0,1]$ (the result at $\beta=0,1$ is obtained by continuity of the right-hand side in $\beta$),
	\begin{equation}
		\inf_{\lambda>\overline\lambda(Q)} ~\lambda + v^\top (\lambda I_n -Q)^{-1} v
		\geq (1-\beta^2)\overline\lambda(Q) +  2\beta\|v\|,
	\end{equation}
	as announced.
\end{proof}

\begin{proof}[Proof of \Cref{lem:prettybound2}]
	Consider the linear case first, $v=Lx$ for some matrix $L\neq0$, then
\begin{align}
	\frac{\mathbb E[\|Lx\|]}{\sqrt{\mathbb E[\|Lx\|^2]}}
	&= \mathbb E\!\left[\sqrt{x^\top \frac{L^\top L}{\Tr (L^\top L)} x}\right] \\
	&\geq \inf_{\substack{S\succeq0\\\Tr S=1}} 
	~\mathbb E\!\left[\sqrt{x^\top S x}\right] \\
	&= \mathbb E[|x_1|],
\end{align}
as $S\succeq0\mapsto \sqrt{x^\top S x}$ is concave for each $x$, and the distribution of $x$ is isotropic. As a result, even when $L=0$,
\begin{equation}
	\mathbb E[\|Lx\|] \geq \mathbb E[|x_1|] \sqrt{\mathbb E[\|Lx\|^2]}.
\end{equation}

What happens when there is an offset? We first notice that, by Jensen's inequality,
\begin{equation}
	\mathbb E[\|v\|] \geq \|\mathbb E[v]\| = \|v_0\|,
\end{equation}
then
\begin{equation}
	\mathbb E[\|v_0+Lx\|]
	= \mathbb E\!\left[\frac12\|v_0+Lx\|+\frac12\|-v_0+Lx\|\right],
\end{equation}
since $x$ is symmetric by inversion. Then since $\|.\|$ is convex,
\begin{equation}
	\mathbb E[\|v_0+Lx\|]
	\geq \mathbb E[\|Lx\|]
	\geq \mathbb E[|x_1|] \sqrt{\Tr(L^\top L)}.
\end{equation}

Assume that either $v_0\neq0$ or $L\neq 0$. If
\begin{equation}
	\mathbb E[|x_1|] \sqrt{\Tr(L^\top L)} \geq \|v_0\|,
\end{equation}
then
\begin{align}
	\frac{\mathbb E[|x_1|] \sqrt{\Tr(L^\top L)}}{\sqrt{\mathbb E[\|v_0+Lx\|^2]}}
	&= \frac{\mathbb E[|x_1|] \sqrt{\Tr(L^\top L)}}{\sqrt{\|v_0\|^2 + \Tr(L^\top L)}} \\
	&\geq \frac{\mathbb E[|x_1|]}{\sqrt{1+\mathbb E[|x_1|]^2}}.
\end{align}
Otherwise
\begin{align}
	\frac{\|v_0\|}{\sqrt{\mathbb E[\|v_0+Lx\|^2]}}
	&= \frac{\|v_0\|}{\sqrt{\|v_0\|^2 + \Tr(L^\top L)}} \\
	&\geq \frac{\mathbb E[|x_1|]}{\sqrt{1+\mathbb E[|x_1|]^2}}.
\end{align}
All in all, when $v$ is an affine function of $x$,
\begin{equation}
	\mathbb E[\|v\|]
	\geq \kappa \sqrt{\mathbb E[\|v\|^2]},
\end{equation}
as claimed.
\end{proof}

When $\nu$ is unidimensional and Gaussian, \Cref{lem:gaussian1d} refines the result of \Cref{lem:prettybound2}. We present its proof now.

\begin{proof}[Proof of \Cref{lem:gaussian1d}]
	The upper bound is a mere application of Jensen's inequality, so the crux is to prove the converse bound. The result is trivial when $b=0$, consider thus $b\neq0$. We may further rescale the problem by $b$ so that we merely need to solve the case $b=1$. Finally since $\nu$ is symmetric, we only really need to solve the case $a\geq0$.
	
	With these reductions in hand, we compute
	\begin{equation}
		\mathbb E[|a + x|]
		= \sqrt{\frac2\pi} \left(a \int_0^a e^{-\nicefrac{x^2}2} \,\mathrm dx + e^{-\nicefrac{a^2}2} \right).
	\end{equation}
	It only remains to show that for $a>0$,
	\begin{equation}
		a \int_0^a e^{-\nicefrac{x^2}2} \,\mathrm dx + e^{-\nicefrac{a^2}2} > \sqrt{1+a^2}.
	\end{equation}
	The derivative of
	\begin{equation}
		f(a) = \int_0^a e^{-\nicefrac{x^2}2} \,\mathrm dx + \frac{e^{-\nicefrac{a^2}2}}a - \frac{\sqrt{1+a^2}}a
	\end{equation}
	with respect to $a>0$ is
	\begin{align}
		f'(a) 
		&= \frac1{a^2} \left( \frac1{\sqrt{1 + a^2}} - e^{-\nicefrac{a^2}2}\right) \\
		&= \frac1{a^2} \left( \frac1{\sqrt{1 + a^2}} - \frac1{\sqrt{e^{a^2}}}\right) 
		> 0.
	\end{align}
	Expanding the exponential and the square root around $0$, $f(a) = O_0(a)$, thus $f$ has limit $0$ at $0^+$. As a result, $f(a)>0$ for all $a>0$.
\end{proof}

\begin{proof}[Proof of \Cref{lem:posbound}]
	Fix $\beta\in[0,1]$. We informally call (A) and (B) the two terms (the second one rid of its multiplicative factor $\gamma(\beta)$). We first prove the stated inequality for $f,\bar\lambda>0$ and $Q_{22}\succ0$, then conclude by continuity. The advantage of setting $f,\bar\lambda>0$ is that both (A) and (B) are positive and so we may solve the informal program
	\begin{equation}
		\sup_{Q,l,r,\Sigma} ~\frac{\text{(A)}}{\text{(B)}}
	\end{equation}
	to retrieve $\gamma(\beta)$.
	
	First of all, we may assume that $Q_{11} = Q_{12}Q_{22}^{-1}Q_{21}$ and $r=0$, since any larger value provides a positive offset on both (A) and (B). Let us define
	\begin{align}
		K &= Q_{22}^{\nicefrac12} \succ 0 \\
		\lambda_0 &= K^{-1} (Q_{21}l_1+Q_{22}l_2) \\
		J &= K^{-1}Q_{21} \\
		L &= J+K \\
		\Gamma &= C^\top K,
	\end{align}
	so that
	\begin{align}
		c &= \|\lambda_0\|^2 + \Tr J^\top J \\
		D &= L^\top L - J^\top J \\
		\bar\lambda &= \bar\lambda(\Gamma^\top \Gamma)\\
		E &= 4 L^\top \Gamma^\top \Gamma L \\
		f &= 4\|\Gamma\lambda_0\|^2.
	\end{align}
	With this notation in hand,
	\begin{align}
		c + \Tr(D\Sigma) 
		&= \Tr (J^\top J(I_n-\Sigma)) + \|\lambda_0\|^2 + \Tr(L\Sigma L^\top ) \\
		&\leq \|\lambda_0\|^2 + \Tr(L\Sigma L^\top ) \\
		\sqrt{f+\Tr(E\Sigma)} 
		&= 2\sqrt{\lambda_0^\top\Gamma^\top\Gamma\lambda_0+\Tr( \Gamma^\top \Gamma L\Sigma L^\top)} \\
		&\leq 2 \sqrt{\bar\lambda(\Gamma^\top \Gamma)} \sqrt{\|\lambda_0\|^2+\Tr( L\Sigma L^\top)}.
	\end{align}
	Call $s = \sqrt{\bar\lambda(\Gamma^\top \Gamma)} > 0$ and $t = \sqrt{\|\lambda_0\|^2+\Tr( L\Sigma L^\top)}> 0$, so that the above becomes
	\begin{align}
		c + \Tr(D\Sigma) 
		&\leq t^2 \\
		\sqrt{f+\Tr(E\Sigma)} 
		&\leq 2 st.
	\end{align}
	Since $\beta\kappa\leq1$,
	\begin{equation}
		\frac{\text{(A)}}{\text{(B)}}
		\leq \frac{t^2+s^2+2st}{t^2+(1-\beta^2)s^2+2\beta\kappa st},
	\end{equation}
	and in turn,
	\begin{equation}
		\sup_{Q,l,r,\Sigma} ~\frac{\text{(A)}}{\text{(B)}}
		\leq \sup_{\zeta>0} ~\frac{\zeta^2+1+2\zeta}{\zeta^2+1-\beta^2+2\beta\kappa\zeta}.
	\end{equation}
	This is actually an equality, one needs to pick the parameters adequately to reproduce any $\zeta$, but we only need to prove the inequality. The derivative of the expression in $\zeta$ is directly proportional to
	\begin{equation}
		1 - \beta^2 - \beta \kappa - (1-\beta \kappa) \zeta.
	\end{equation}
	If $\beta$ is small enough that $1 - \beta^2 - \beta \kappa > 0$, then the maximum occurs at
	\begin{equation}
		\frac{1 - \beta^2 - \beta \kappa}{1-\beta \kappa},
	\end{equation}
	with value
	\begin{equation}
		\frac{2 - \beta^2 - 2 \beta \kappa}{1 - \beta^2 (1 + \kappa^2)}.
	\end{equation}
	Otherwise, the expression is decreasing in $\zeta$, so that the supremum arises at $\zeta\to0^+$, with value
	\begin{equation}
		\frac1{1-\beta^2}.
	\end{equation}
	This establishes that
	\begin{equation}
		(A) \leq \gamma(\beta) (B)
	\end{equation}
	for all parameters such that $f,\bar\lambda>0$ and $Q_{22}\succ0$. By continuity of (A) and (B), this also stands when this positivity assumption is relaxed. This proves the first part of the statement. 
	
	As for $\gamma$, we let $\tilde\beta\in(0,1)$ be uniquely defined by the equation $1 - \tilde\beta^2 - \tilde\beta \kappa = 0$. On $[\tilde\beta,1]$, $\gamma$ is nondecreasing and on $[0,\tilde\beta)$, its derivative is directly proportional to
	\begin{equation}
		- \kappa + \beta (1 + 2 \kappa^2) - \beta^2 (\kappa + \kappa^3).
	\end{equation}
	This quadratic in $\beta$ vanishes at $\bar\beta=\nicefrac{\kappa}{1+\kappa^2} < \tilde\beta$ and at $\nicefrac1\kappa>\tilde\beta$, therefore $\gamma$ reaches a minimum at $\bar\beta$ with value
	\begin{equation}
		\bar\gamma = \gamma(\bar\beta) = 1+\frac1{1+\kappa^2},
	\end{equation}
	establishing the second part of the statement.
\end{proof}

\subsection{Study of $\bar\gamma$}
\label{sec:gamma}

The ratio $\bar\gamma$ depends on the prior distribution, and can never fall below $\upsilon_n$, the ratio obtained for the uniform distribution on the sphere (of radius $\sqrt n$). As a result, $\upsilon_n$ provides an upper bound on the tightness of the approximation \eqref{eq:sandwich}. On the other hand, $\bar\gamma$ could be as large as $2$, which means there are priors for which \eqref{prog:optimistic} is not much more informative than \eqref{prog:genoptimistic} in the worst case. However, for Gaussian priors, $\bar\gamma$ is independent of the dimension and approximately equals $1.72$. More precisely, we present the following proposition.

\begin{prop}
\label{prop:gamma}
	For any isotropic prior of covariance $I_n$,
	\begin{equation}
		\bar\gamma \geq \upsilon_n \triangleq 
		\frac32 + \frac12 \frac1{1+\frac{2n\Gamma(\nicefrac n2)^2}{\pi\Gamma(\nicefrac{(n+1)}2)^2}},
	\end{equation}
	with equality if and only if the prior is the uniform distribution on the sphere of radius $\sqrt n$. The sequence $(\upsilon_n)$ increases with limit
	\begin{equation}
		\upsilon_\infty = 
		\frac{2(3+\pi)}{4+\pi} \approx 1.72,
	\end{equation}
	which is the ratio $\bar\gamma$ for Gaussian priors, regardless of the dimension.
\end{prop}

To prove \Cref{prop:gamma}, we first establish a formula for $\mathbb E[|x_1|]$ in term of $\mathbb E[\|x\|]$. Later we will use the Cauchy-Schwarz inequality
\begin{equation}
	\mathbb E[\|x\|]^2 \leq \mathbb E[\|x\|^2],
\end{equation}
which is an equality if and only if $\|x\|$ is constant.

\begin{lemma}
\label{lem:x1}
	When $\nu$ is isotropic,
	\begin{equation}
		\mathbb E[|x_1|] = 
		\frac{\Gamma(\nicefrac n2)}{\sqrt\pi\Gamma(\nicefrac{n+1}2)}
		\mathbb E[\|x\|].
	\end{equation}
\end{lemma}

\begin{proof}[Proof of \Cref{lem:x1}]
	The key idea is to notice the formula is ``homogeneous,'' i.e., both sides are linear in $\nu$, the distribution of $x$. Then it suffices to prove it for say $\nu=\mathcal N(0,I_n)$, then ``integrating'' the formula to retrieve any $\nu$. This second step involves some measure theory, specifically the disintegration theorem, for which we advise consulting \cite{disintegration}.
	
	When $\nu=\mathcal N(0,I_n)$ the formula holds for $x_1 \sim \mathcal N(0,1)$ and so
	\begin{equation}
		\mathbb E[|x_1|] = \sqrt{\frac2\pi},
	\end{equation}
	and for $\|x\|\sim\chi(n)$ (the chi distribution with $n$ degrees of freedom) which has mean
	\begin{equation}
		\mathbb E[\|x\|] = \sqrt2 \frac{\Gamma(\nicefrac{n+1}2)}{\Gamma(\nicefrac n2)}.
	\end{equation}
		
	Consider now $\nu$ isotropic, we may express it
	\begin{equation}
		\mathrm d\nu(x) = \mathrm d\eta(\|x\|) \mathrm d\nu_{\|x\|},
	\end{equation}
	where $\nu_{\|x\|}$ is the uniform probability distribution on the sphere of radius $\|x\|$ and $\eta$ is the distribution of $\|x\|$. With this disintegration,
	\begin{align}
		\mathbb E[|x_1|]
		&= \int_0^\infty \int_{\|x\| \mathbb S^{n-1}} |x_1| \mathrm d\nu_{\|x\|} \mathrm d\eta(\|x\|) \\
		&= \int_0^\infty \aleph
		\|x\| \mathrm d\eta(\|x\|) \\
		&= \aleph 
		\mathbb E[\|x\|],
	\end{align}
	where
	\begin{equation}
		\aleph=\int_{\mathbb S^{n-1}} |x_1| \mathrm d\nu_{1}>0.
	\end{equation}
	Using this formula with $\nu=\mathcal N(0,I_n)$, we obtain
	\begin{equation}
		\aleph = \frac{\Gamma(\nicefrac n2)}{\sqrt\pi\Gamma(\nicefrac{n+1}2)}
	\end{equation}
	which seals the proof.
\end{proof}

\begin{proof}[Proof of \Cref{prop:gamma}]
	When $x \sim \mathcal N(0,I_n)$, $x_1 \sim \mathcal N(0,1)$ is a scalar Gaussian random variable, therefore
\begin{equation}
	\mathbb E[|x_1|] 
	= \int_{-\infty}^\infty |t| \frac{e^{-\nicefrac{t^2}2}}{\sqrt{2\pi}} \,\mathrm dt
	= \sqrt{\frac2\pi},
\end{equation}
and so follows the value of $\bar\gamma$ for Gaussian priors.

	For any isotropic prior of covariance $I_n$, using \Cref{lem:x1}, we may express
	\begin{equation}
		\bar\gamma 
		= \frac32 + \frac12
		\frac1{1+\frac{2\Gamma(\nicefrac n2)^2}{\pi\Gamma(\nicefrac{n+1}2)^2}
		\mathbb E[\|x\|]^2}
	\end{equation}
	By Cauchy-Schwarz inequality,
	\begin{equation}
		\bar\gamma \geq \upsilon_n
	\end{equation}
	with equality if and only if $\|x\|$ is constant, that is if and only if the prior is spherical.
	
	Finally to analyze the monotonicity and limit of $(\upsilon_n)$ we define for $x>0$
	\begin{equation}
		u(x) = 2\ln \frac{\Gamma(\nicefrac{x+1}2)}{\sqrt x\Gamma(\nicefrac x2)}.
	\end{equation}
	Its derivative is
	\begin{equation}
		u'(x) = \psi(\nicefrac{x+1}2) - \psi(\nicefrac x2) - \frac1x > 0
	\end{equation}
	where $\psi$ is the digamma function and where the positivity ensues from Theorem 7 (with $n=0$, $s=\nicefrac12$ and $x$ substituted with $\nicefrac{x-1}{2}$) of \cite{Alzer}. In turn,
	\begin{equation}
		\upsilon_n
		=\frac32 + \frac12 \frac1{1+\frac2\pi e^{-u(n)}}
	\end{equation}
	increases with $n$. We also remark that
	\begin{equation}
		u(x) + u(x+1) = \ln \frac{x}{4(x+1)},
	\end{equation}
	so $u(x)$ tends to $-\ln2$ as $x$ goes to infinity, and thus
	\begin{equation}
		\upsilon_n \to_n \upsilon_\infty,
	\end{equation}
	where $\upsilon_\infty$ is $\bar\gamma$ when the prior is Gaussian.
\end{proof}

\section{Monotonicity of rank}

\begin{proof}[Proof of \Cref{prop:orth}]
	In a first step, we proceed to a reduction and evacuate a particular case. First of all, let $X$ be a solution of \eqref{prog:pessimistic} and let
	\begin{equation}
		X^* = P^{<0}_D X P^{<0}_D.
	\end{equation}
	Then $\rk X^* \leq \rk X$ and $X^*$ is a solution as well since
	\begin{align}
		\Tr(EX^*) &= \Tr(P^{<0}_D E P^{<0}_D X) \leq \Tr(EX)\\
		\Tr(DX^*) &= \Tr(-D^- X) \leq \Tr(DX).
	\end{align}
	For this reason, we may restrict the ambient space to $\Ima P^{<0}_D$ when studying minimal rank solutions. This is equivalent to assuming $D \prec 0$, which we do in the remainder of this proof.
	
	Second, if $0$ is a solution of \eqref{prog:pessimistic}, then it is the only solution of minimal rank, and of course it is an orthogonal projection matrix. In the remainder, we assume that $0$ is not a solution. Note that if $\Tr(EX) = 0$, the objective at $X$ is at least as large as that at $0$, hence any solution $X$ must be such that $\Tr(EX) > 0$. We may thus focus on arguments $X$ for which $\Tr(EX)>0$, call $\mathcal D$ this domain. 
		
	\medskip
	In a second step, we characterize solutions of \eqref{prog:pessimistic}. Since the objective (which we shall denote $g$) is concave on $\mathcal D$ convex, the program is concave. The objective is smooth on $\mathcal D$, with gradient
	\begin{equation}
		\nabla g(X) = D + \frac E{2\sqrt{f+\Tr(EX)}}.
	\end{equation}
	In this case then, solutions are easily characterized. Let $X\in\mathcal D$ be a solution, $Y\in\mathcal D$ and define for all $t\in[0,1]$,
	\begin{equation}
		h(t) = g(X+t(Y-X)).
	\end{equation}
	Since $X$ is a solution and $\mathcal D$ is convex, $h$ continuously differentiable reaches a minimum at $0$, hence
	\begin{equation}
		h'(0) = \nabla g(X)^\top (Y-X)\geq0.
	\end{equation}	
	By continuity, this also applies to all $Y \in \bar{\mathcal D} = \mathcal S = \{0\preceq X \preceq I_n\}$. Therefore $X$ solves
	\begin{equation}
		\min_{Z\in\mathcal S} ~\Tr(\nabla g(X)^\top Z)
	\end{equation}
	and so,
	\begin{equation}
		P_{\nabla g(X)}^{<0} \preceq X \preceq P_{\nabla g(X)}^{\leq0}.
	\end{equation}
	
	Let now $X$ be a solution of minimal rank. For all $Z$ such that
	\begin{equation}
		P_{\nabla g(X)}^{<0} \preceq Z \preceq P_{\nabla g(X)}^{\leq0},
	\end{equation}
	we have
	\begin{equation}
		\Tr(\nabla g(X)^\top (Z-X))=0.
	\end{equation}
	Hence, for these $Z$, $\Tr(DZ)$ can be rewritten as a simple function of $\Tr(EZ)$. As $X$ solves \eqref{prog:pessimistic}, it is also a solution of the same program restricting the constraint set, namely it solves
	\begin{equation}
		\min_{P_{\nabla g(X)}^{<0} \preceq Z \preceq P_{\nabla g(X)}^{\leq0}}
		~-\frac{\Tr(EZ)}{2\sqrt{f+\Tr(EX)}} 
		+\sqrt{f+\Tr(EZ)}.
	\end{equation}
	The objective is strictly concave in $\Tr(EZ)$ and $E\succeq0$, thus this program is solved at $P_{\nabla g(X)}^{<0}$ or $P_{\nabla g(X)}^{\leq0}$. If the first argument is not a solution, however, we run into a contradiction. Indeed, then
	\begin{equation}
		\Tr(EX) = \Tr(EP_{\nabla g(X)}^{\leq0}),
	\end{equation}
	and so,
	\begin{align}
		&-\frac{\Tr(EP_{\nabla g(X)}^{<0})}{2\sqrt{f+\Tr(E P_{\nabla g(X)}^{\leq0})}} 
		+\sqrt{f+\Tr(EP_{\nabla g(X)}^{<0})}\\
		&\quad>
		-\frac{\Tr(EP_{\nabla g(X)}^{<0})}{2\sqrt{f+\Tr(E P_{\nabla g(X)}^{\leq0})}} 
		+\sqrt{f+\Tr(EP_{\nabla g(X)}^{<0})}.
	\end{align}
	Rearranging the terms yields
	\begin{align}
		&2\sqrt{f+\Tr(EP_{\nabla g(X)}^{<0})}\sqrt{f+\Tr(E P_{\nabla g(X)}^{\leq0})}\\
		&\quad>
		\Tr(EP_{\nabla g(X)}^{<0})+\Tr(EP_{\nabla g(X)}^{<0})+2f,
	\end{align}
	which is a contradiction. As a a result, $P_{\nabla g(X)}^{<0}$ is a solution of \eqref{prog:pessimistic}. This, added to the facts that $X\succeq P_{\nabla g(X)}^{<0}$ and that $X$ has minimal rank, implies that $X = P_{\nabla g(X)}^{<0}$. Therefore $X$ is an orthogonal projection matrix.
\end{proof}

\begin{proof}[Proof of \Cref{thm:monotonicity}]
	As in the previous proof, we may assume that $D\prec0$. Moreover, the result holds immediately if $X_2 = 0$ or if $\epsilon_1=\epsilon_2$, we thus assume that $\epsilon_1<\epsilon_2$ and $X_2\neq0$. Let us parametrize the hypotheses more succintly:
	\begin{equation}
		E = \epsilon^2 E_0, ~f = \epsilon^2 f_0,
	\end{equation}
	with $\epsilon\geq0$ varying. We denote
	\begin{equation}
		R_a = P_{D+aE_0}^{<0}.
	\end{equation}
	Since $D+aE_0$ increases with $a$, the dimension of its negative eigenspace decreases with $a$, which is none else than $\rk R_a$.
	
	We first show that
	\begin{equation}
		\Tr(DX_1) \leq \Tr(DX_2) < 0,
	\end{equation}
	directly implying that $X_1\neq0$. The second inequality is immediate as $X_2\neq0$ and $D\prec0$. Regarding the first one, since $X_1$ is a solution with hypothesis $\epsilon=\epsilon_1$, and $X_2$ is a solution under the second hypothesis,
	\begin{align}
		&\Tr(DX_1) + \epsilon_1 \sqrt{f_0 + \Tr(E_0X_1)} \\
		&\quad\leq\Tr(DX_2) + \epsilon_1 \sqrt{f_0 + \Tr(E_0X_2)} \\
		&\quad=\left(1-\frac{\epsilon_1}{\epsilon_2}\right)\Tr(DX_2) \\
		&\quad\phantom{=}+ \frac{\epsilon_1}{\epsilon_2} \left(\Tr(DX_2)+\epsilon_2\sqrt{f_0 + \Tr(E_0X_2)}\right) \\
		&\quad\leq\left(1-\frac{\epsilon_1}{\epsilon_2}\right)\Tr(DX_2) \\
		&\quad\phantom{=}+ \frac{\epsilon_1}{\epsilon_2} \left(\Tr(DX_1)+\epsilon_2\sqrt{f_0 + \Tr(E_0X_1)}\right),
	\end{align}
	therefore
	\begin{equation}
		\left(1-\frac{\epsilon_1}{\epsilon_2}\right) (\Tr(DX_2) - \Tr(DX_1)) \geq 0.
	\end{equation}
	This fact also helps us show that
	\begin{equation} \label{eq:monotonicity}
	\begin{aligned}
		&\epsilon_2 \sqrt{f_0 + \Tr(E_0X_2)} \\
		&\quad=\epsilon_2 \sqrt{f_0 + \Tr(E_0X_2)} + \Tr(DX_2) - \Tr(DX_2) \\
		&\quad\leq\epsilon_2 \sqrt{f_0 + \Tr(E_0X_1)} + \Tr(DX_1) - \Tr(DX_2) \\
		&\quad\leq\epsilon_2 \sqrt{f_0 + \Tr(E_0X_1)}.
	\end{aligned}
	\end{equation}

	In the present case ($X_1,X_2\neq0$), we have characterized the solutions in the proof of \Cref{prop:orth}:
	\begin{equation}
		X_1 = R_{a_1}, ~X_2 = R_{a_2},
	\end{equation}
	where,
	\begin{align}
		a_1 &= \frac{\epsilon_1}{2\sqrt{f_0+\Tr(E_0X_1)}}\\
		a_2 &= \frac{\epsilon_2}{2\sqrt{f_0+\Tr(E_0X_2)}}.
	\end{align}
	
	Given the monotonicity we have derived earlier in \eqref{eq:monotonicity}, $a_1\leq a_2$,	and as a result,
	\begin{equation}
		\rk X_1 = \rk R_{a_1} \geq \rk R_{a_2} = \rk X_2,
	\end{equation}
	which terminates the proof.
\end{proof}

\begin{proof}[Proof of \Cref{coro:monotonicity}]
	Observe that \eqref{prog:Bayes_linquad} and \eqref{prog:optimistic} correspond to the Pessimistic Program, \eqref{prog:pessimistic}, with respective hypothesis $0 CC^\top$ and $\bar\gamma^2 CC^\top$ in lieu of $CC^\top$. \Cref{thm:monotonicity} then guarantees this hierarchy of minimal ranks.
\end{proof}

\begin{proof}[Proof of \Cref{cor:noinf_exp}]
	If $D\succeq0$, $\Sigma=P^{<0}_D=0$ is a solution of the Bayesian Program. In turn, $\Sigma=0$ is a solution of the Universal Optimistic Program, since it only differs from the Bayesian Program by a constant in the objective. Moreover, \Cref{coro:monotonicity} implies that the minimal rank of a solution of \eqref{prog:pessimistic} and \eqref{prog:optimistic} is $0$.
\end{proof}

\section{Numerical solution}
\label{sec:numerical_app}

\subsection{Properties of $h$}

\begin{proof}[Proof of \Cref{prop:line}]
	The motivation behind the definition of $h$ comes from the following rewriting
	\begin{align}
		&\min_{0\preceq X\preceq I_n} ~\Tr(DX) + \sqrt{f+\Tr(EX)} \\
		&\quad=\min_{\substack{0\preceq X\preceq I_n\\t\geq\Tr(EX)}} ~\Tr(DX) + \sqrt{f+t} \\
		&\quad=\min_{t\geq0} ~h(t) + \sqrt{f+t}.
	\end{align}
	This directly establishes the second part of the statement. 
	
	Assume that $Y$ solves \eqref{prog:pessimistic}. Since both programs share the same value,
	\begin{align}
		\min_{t\geq0} ~h(t) + \sqrt{f+t} 
		&= \Tr(DY) + \sqrt{f+\Tr(EY)} \\
		&\geq h(\Tr(EY)) + \sqrt{f+\Tr(EY)}.
	\end{align}
	The inequality is therefore an equality, therefore $Y$ solves the program defining $h(\Tr(EY))$, and $\Tr(EY)$ solves \eqref{prog:line}.
	
	Assume now the converse, $Y$ solves the program defining $h(\Tr(EY))$, and $\Tr(EY)$ solves \eqref{prog:line}. Again, since both programs share the same value, and since $Y$ solves the program defining $h(\Tr(EY))$,
	\begin{align}
		&\min_{0\preceq X\preceq I_n} ~\Tr(DX) + \sqrt{f+\Tr(EX)} \\
		&\quad= h(\Tr(EY)) + \sqrt{f+\Tr(EY)} \\
		&\quad= \Tr(DY) + \sqrt{f+\Tr(EY)}.
	\end{align}
	As a result, $Y$ solves \eqref{prog:pessimistic}.
\end{proof}

\Cref{prop:suboptim} relies on the following lemma. Observe that the difference between the two bounds is directly controlled by $a,b$, independently of $h$. 

\begin{lemma} \label{lem:hproperties}
	The function $h$ is continuous and convex, decreasing on $[0,\bar t]$ and constant on $[\bar t, \infty)$. In addition, for any $0\leq a < b$,
	\begin{equation}
		h(b) + \sqrt{f+a} \leq \min_{t\in[a,b]} ~h(t) + \sqrt{f+t} \leq h(b) + \sqrt{f+b}.
	\end{equation}
\end{lemma}

\begin{proof}[Proof of \Cref{lem:hproperties}]
	Continuity is a direct consequence of the minimum theorem: the objective does not depend on the parameter $t$, whereas the domain is a non-empty compact-valued continuous correspondence in $t$. Nonincreasingness comes directly from the fact that this correspondence is nondecreasing and the objective is minimized. 
	
	Regarding convexity, let $u,v\geq0$ and $\lambda\in[0,1]$. Let then $X$ solve the program that defines $h(u)$ and $Y$ solve the program that defines $h(v)$. Then $\lambda X + (1-\lambda) Y$ satisfies the constraint that defines $h(\lambda u + (1-\lambda)v)$, so its value must be at least as large as $h(\lambda u + (1-\lambda)v)$, namely
	\begin{equation}
		\lambda h(u) + (1-\lambda) h(v) \geq h(\lambda u + (1-\lambda)v).
	\end{equation}
	
	Finally, $P_D^{<0}$ solves
	\begin{equation}
		\min_{0\preceq X\preceq I_n} ~\Tr(DX),
	\end{equation}
	and we have let $\bar t=\Tr(EP_D^{<0})$. Since $P_D^{<0}$ solves the program without the trace constraint, it solves the program defining $h(t)$ whenever $t\geq \bar t$, therefore $h(t) = h(\bar t)$ for all $t\geq \bar t$. Furthermore, \Cref{lem:trsols} guarantees all other solutions $Y$ of this SDP satisfy $Y\succeq P_D^{<0}$, and in particular $\Tr(EY) \geq \bar t$. As a result, for all $0\leq t<\bar t$, $h(t) > h(\bar t)$, and since $h$ is convex this implies that $h$ is actually strictly decreasing on $[0,\bar t]$.
	
	Regarding the two bounds, the first one relies on the monotonicity of $h$, whereas the second one is simply obtained by setting $t=b$. 
\end{proof}

\begin{proof}[Proof of \Cref{prop:suboptim}]
It is rather immediate to see that,
\begin{align}
	&\min_{t\geq0} ~h(t) + \sqrt{f+t} \\
	&\quad=\min_{t\in[0,\bar t]} ~h(t) + \sqrt{f+t} \\
	&\quad=\min_{0\leq n< N} ~\min_{t\in[u_n,u_{n+1}]} ~h(t) + \sqrt{f+t} \\
	&\quad\geq\min_{0\leq n< N} ~h(u_{n+1}) + \sqrt{f+u_n} \\
	&\quad\geq\min_{0\leq n< N} ~h(u_{n+1}) + \sqrt{f+u_{n+1}} - \rho.
\end{align}
The first equality is obtained as $h$ is constant on $[\bar t, \infty)$. The second one comes from the fact that
\begin{equation}
	\bigcup_{0\leq n< N} [u_n, u_{n+1}]  \supset [0,\bar t].
\end{equation}
The first inequality is directly lifted from \Cref{lem:hproperties}.
\end{proof}

\eqref{prog:spop} can receive a similar treatment to that of \eqref{prog:pessimistic}. As mentioned earlier, it first relies on resolving the inner maximization: for $\zeta\geq0$,
\begin{equation}
	\max_{\beta\in[0,1]} ~1-\beta^2 + \beta\zeta = \begin{cases}
		\zeta &\text{if } \zeta\geq2\\
		1+\frac{\zeta^2}4 &\text{otherwise.}
	\end{cases}	
\end{equation}
It is notable that this expression is concave in $\zeta^2$, since the above formula is used with $\zeta^2 = \nicefrac{\kappa^2(f+\Tr(E\Sigma))}{\bar\lambda^2}$ in resolving the inner maximization of \eqref{prog:spop}. In turn, even though the lower bound of \Cref{thm:boundproj} is only obtained for $\Sigma$ orthogonal projection matrix, it is concave in $\Sigma$ hence there is no loss of generality considering all $0\preceq\Sigma\preceq I_n$, \eqref{prog:spop} is solved by an extreme point, i.e., an orthogonal projection matrix.

The second step separates both cases, \eqref{prog:spop} is the minimum of the two following programs:
\begin{equation}
	\begin{aligned}
		\min_{0\preceq\Sigma\preceq I_n} &\Tr(D\Sigma) + c + \kappa\sqrt{f+\Tr(E\Sigma)} \\
		\text{s.t.} &\Tr(E\Sigma) \geq \check t
	\end{aligned}
\end{equation}
and
\begin{equation}
	\begin{aligned}
		\min_{0\preceq\Sigma\preceq I_n} &\Tr(\check D\Sigma) + c + \bar\lambda+\frac{\kappa^2 f}{4\bar\lambda} \\
		\text{s.t.} &\Tr(E\Sigma) \leq \check t,
	\end{aligned}
\end{equation}
where we have let
\begin{equation}
	\check t = \frac{4\bar\lambda^2}{\kappa^2} - f, ~\check D = D + \frac{\kappa^2}{4\bar\lambda} E.
\end{equation}
Moreover, $\Sigma$ which solves the program of smallest value also solves \eqref{prog:spop}. The first program is akin to \eqref{prog:pessimistic} with an additional constraint on $\Tr(E\Sigma)$, which leads to the definition of a different function, $\check h$, whose properties follow along the line of \Cref{lem:hproperties} (save for its domain of strict monotony), which ultimately leads to an analogous grid search. The second program is a simple SDP, relatively inexpensive to solve.

\subsection{Coherence}

\begin{proof}[Proof of \Cref{prop:coherence}]
	Consider $t \in (0,\bar t)$ where $\bar t$ is the threshold after which $h$ is constant. The program that defines $h(t)$ is convex and satisfies Slater's condition, therefore $0\preceq X\preceq I_n$ is a solution if and only if there exist $\lambda\geq0$, $M_1,M_2\succeq0$ such that
	\begin{equation}
		D + \lambda E - M_1 + M_2 = 0,
	\end{equation}
	and $\Tr(M_1X) = \Tr(M_2(I-X)) = \lambda(\Tr(EX)-t) = 0$. Moreover, since $h$ is strictly decreasing, the constraint on $\Tr(EX)$ must be active, so $\Tr(EX) = t$. Once $\lambda$ is fixed, all the other conditions are equivalent to $X$ solving the KKT conditions of the following convex program,
	\begin{equation}
	\label{eq:lambda}
		\min_{0 \preceq X \preceq I_n} ~\Tr((D+\lambda E) X).
	\end{equation}
	This program also satisfies Slater's condition, therefore $X$ is a solution of the program defining $h(t)$ if and only if $\Tr(EX) = t$ and there exists $\lambda\geq0$ such that
	\begin{equation}
		P_{D+\lambda E}^{<0} \preceq X \preceq P_{D+\lambda E}^{\leq0}.
	\end{equation}
	Note that $\lambda=0$ is not a possibility, otherwise $X \succeq P_D^{<0}$ and so $\Tr(EX) \geq \bar t > t$. All in all, $X$ is a solution of the program defining $h(t)$ if and only if $\Tr(EX) = t$ and there exists $\lambda>0$ such that
	\begin{equation}
		P_{D+\lambda E}^{<0} \preceq X \preceq P_{D+\lambda E}^{\leq0}.
	\end{equation}

	We now prove that for $\lambda>0$, there is at most one $X$ such that the above condition is satisfied. If $P_{D+\lambda E}^{<0}=P_{D+\lambda E}^{\leq0}$, surely $X = P_{D+\lambda E}^{<0}$ is the only possible solution. Otherwise, since $\rk(D+\lambda E) \geq n-1$, the difference in rank between the two projections is exactly $1$, we may let $u$ be a unit-vector such that
	\begin{equation}
		P^{\leq0}_{D+\lambda E} = P^{<0}_{D+\lambda E}+uu^\top.
	\end{equation}
	In this case, if ever
	\begin{equation}
		\Tr(EP^{<0}_{D+\lambda E}) = \Tr(EP^{\leq0}_{D+\lambda E}),
	\end{equation}
	we would have $u^\top Eu = 0$ and $(D + \lambda E)u = 0$, thus $Du=Eu=0$, thereby contradicting the assumption that $\ker D \cap \ker E = \{0\}$. Still in this case then, the only possible solution is the unique convex combination $X$ of $P^{<0}_{D+\lambda E}, P^{\leq0}_{D+\lambda E}$ (if it even exists) such that $\Tr(EX) = t$.
	
	All in all, this analysis reveals that $\lambda$ corresponds to a solution $X$ if and only if
	\begin{equation}
		\Tr(EP^{<0}_{D+\lambda E}) \leq t \leq \Tr(EP^{\leq0}_{D+\lambda E}),
	\end{equation}
	and moreover the solution $X$ is unique with $\lambda$ given. It also reveals that solutions are convex combination of at most two orthogonal projection matrices.

	With this characterization in hand, we may focus on $\lambda$. We first show that for all $\lambda_1<\lambda_2$,
	\begin{equation}
		\Tr(EP_{D+\lambda_1 E}^{<0}) \geq \Tr(EP_{D+\lambda_2 E}^{\leq0}).
	\end{equation}
	Since the projections solve \eqref{eq:lambda} at $\lambda_1,\lambda_2$ respectively,
	\begin{align}
		\Tr((D+\lambda_1 E)P^{<0}_{D+\lambda_1 E} ) 
		&\leq \Tr((D+\lambda_1 E)P^{\leq0}_{D+\lambda_2 E}) \\
		\Tr((D+\lambda_2 E)P^{\leq0}_{D+\lambda_2 E} ) 
		&\leq \Tr((D+\lambda_2 E)P^{<0}_{D+\lambda_1 E}),
	\end{align}
	in particular,
	\begin{align}
		&\lambda_1(\Tr( EP^{\leq0}_{D+\lambda_2 E} )-\Tr( EP^{<0}_{D+\lambda_1 E} )) \\
		&\quad\geq \Tr(DP^{\leq0}_{D+\lambda_2 E}) - \Tr(DP^{<0}_{D+\lambda_1 E}) \\
		&\quad\geq \lambda_2(\Tr( EP^{\leq0}_{D+\lambda_2 E} )-\Tr( EP^{<0}_{D+\lambda_1 E} )),
	\end{align}
	and thus, as claimed,
	\begin{equation}
		\Tr( EP^{<0}_{D+\lambda_1 E} ) \geq \Tr( EP^{\leq0}_{D+\lambda_2 E}).
	\end{equation}

	Let now $X_1\neq X_2$ be two solutions, they correspond to $\lambda_1<\lambda_2$ (without loss of generality). Using the above result and the characterization in terms of $\lambda$,
	\begin{align}
		\Tr( EP^{\leq0}_{D+\lambda_1 E} ) 
		&\geq t 
		= \Tr( EP^{<0}_{D+\lambda_1 E} ) 
		= \Tr( EP^{\leq0}_{D+\lambda_2 E}) \\
		&\geq \Tr( EP^{<0}_{D+\lambda_2 E}).
	\end{align}
	In turn,
	\begin{equation}
		X_1 = P^{<0}_{D+\lambda_1 E}, ~X_2 = P^{\leq0}_{D+\lambda_2 E}.
	\end{equation}
	Moreover all inequalities of the previous result are equalities, the projections solve each other's program \eqref{eq:lambda} and thus
	\begin{equation}
		P_{D+\lambda_2E}^{<0}
		\preceq P_{D+\lambda_1E}^{<0}
		\preceq P_{D+\lambda_2E}^{\leq0}
		\preceq P_{D+\lambda_1E}^{\leq0}.
	\end{equation}
	We must then have,
	\begin{equation}
		P_{D+\lambda_2E}^{<0}
		= P_{D+\lambda_1E}^{<0}
		\prec P_{D+\lambda_2E}^{\leq0}
		= P_{D+\lambda_1E}^{\leq0},
	\end{equation}
	but this brings a contradiction as
	\begin{equation}
		\Tr(EP_{D+\lambda_1E}^{<0}) = \Tr(EP_{D+\lambda_2E}^{\leq0}) = \Tr(EP_{D+\lambda_1E}^{\leq0}).
	\end{equation}
	Therefore, the solution is unique.
\end{proof}

\section{On Bayesian linear-quadratic persuasion}

\subsection{An important technical lemma}

We had stressed the importance of \Cref{lem:trsols}. On the one hand, it is useful for the Bayesian case, as it solves directly \eqref{prog:Bayes_linquad}. On the other hand, it will prove a helpful tool later on as well, when we discuss the non-Bayesian programs.

\begin{proof}[Proof of \Cref{lem:trsols}]
	One way of obtaining $P_D^{<0},P_D^{\leq0}$ is to diagonalize $D = R \Delta R^\top$ with $R$ a rotation and $\Delta$ a diagonal matrix with decreasing eigenvalues. Explicitly,
	\begin{equation}
		\Delta = \begin{bmatrix} \Delta^- & 0 & 0 \\ 0&0&0\\ 0&0 & \Delta^+ \end{bmatrix},
	\end{equation}
	where some of these diagonal blocks potentially have dimension $0$, and $\Delta^-,\Delta^+$ are definite. Then if $p\leq q$ are the number of negative and non-positive eigenvalues, and $J_r$ is the diagonal matrix with $r$ ones and $n-r$ zeroes in this order,
	\begin{equation}
		P_D^{<0} = RJ_pR^\top, ~P_D^{\leq0} = RJ_qR^\top.
	\end{equation}
	
	Define
	\begin{align}
		D^- &= -P_D^{<0} D \succeq0, \\
		D^+ &= (I-P_D^{\leq0}) D \succeq0,
	\end{align}
	so that $D = D^+ - D^-$. Note that $P_D^{<0},P_D^{\leq0},D$ all commute. No matter $0\preceq X\preceq I_n$,
	\begin{equation}
		\Tr(D^+ X) \geq 0, ~\Tr(D^- X) \leq \Tr(D^-).
	\end{equation}
	At the same time, these are equalities whenever $P_D^{<0}\preceq X\preceq P_D^{\leq0}$, thus all such $X$ are solution of
	\begin{equation}
		\min_{0\preceq X\preceq I_n} ~\Tr(DX).
	\end{equation}

	This condition turns out to be sufficient as well. Indeed, let $X$ be a solution, we must have
	\begin{equation}
		\Tr(D^+ X) = 0, ~\Tr(D^- X) = \Tr(D^-).
	\end{equation}
	Since $\Delta^+, \Delta^-$ are definite, this implies that $X$ takes the general form
	\begin{equation}
		X = R \begin{bmatrix} I_p&\star&\star\\\star&\star&\star\\\star&\star&0 \end{bmatrix} R^\top,
	\end{equation}
	where $\star$ are any block. Since $X\succeq0$, we must rather have
	\begin{equation}
		X = R \begin{bmatrix} I_p&\star&0\\\star&\star&0\\0&0&0 \end{bmatrix} R^\top,
	\end{equation}
	and since $I_n-X\succeq0$, we must have
	\begin{equation}
		X = R \begin{bmatrix} I_p&0&0\\0&\star&0\\0&0&0 \end{bmatrix} R^\top,
	\end{equation}
	where $0\preceq\star\preceq I_{q-p}$. All in all, this implies that $P_D^{<0}\preceq X\preceq P_D^{\leq0}$. Moreover, the rank of such $X$ is $p + \rk \star\geq p$ where $\star$ is the center block, with equality if and only if it is $0$. In other words, $X = RJ_pR^\top = P^{<0}_D$ is the unique solution of minimal rank.
\end{proof}

\subsection{About which covariances can be produced}

Before proving \Cref{thm:isotrop}, we first establish a lemma that takes care of most of the proof, and delegates the ``hard part'' to another lemma.

\begin{lemma}
\label{lem:whichcov}
	The following statements are equivalent,
	\begin{enumerate}[(i)]
		\item $\mathcal S = \mathcal S_\nu$;
		\item \eqref{prog:Bayes_linquad_nonisotropic} and \eqref{prog:Bayes_linquad} have same value for all $D$;
		\item[(iv)] for all orthogonal projection matrix $P$, $\mathbb E[x\!\mid\!Px] = Px$.
	\end{enumerate}
\end{lemma}

\begin{proof}[Proof of \Cref{lem:whichcov}]
First of all, $\mathcal S_\nu \subset \mathcal S$ is convex. Indeed, let $t\in[0,1]$ and $\Sigma_1,\Sigma_2\in \mathcal S_\nu$, they correspond to the covariance of two random variables $\hat x_1, \hat x_2$ respectively, which by nature satisfy
\begin{equation}
	\mathbb E[x\!\mid\!\hat x_1] = \hat x_1, ~\mathbb E[x\!\mid\!\hat x_2] = \hat x_2.
\end{equation}
Let $i$ be an independent random variable taking value $1$ with probability $t$ and $2$ with probability $1-t$. Consider the message $y = \hat x_i$ and the estimator it generates,
\begin{equation}
	\hat x
	= \mathbb E[x\!\mid\!y]
	= \mathbb E[\mathbb E[x\!\mid\!y,i]]
	= \mathbb E[y]
	= y,
\end{equation} 
where the outer most expectation is taken with respect to $i$. In other words, from the point of view of a Bayesian agent receiving $y$, either the message was $y=\hat x_1$, in which case the estimator is $y$, or the message was $y=\hat x_2$, in which case the estimator is still $y$. The covariance of $\hat x$ is none other than
\begin{equation}
	\Sigma 
	= \mathbb E[yy^\top]
	= \mathbb E[\mathbb E[\hat x_i \hat x_i^\top]]
	= t\Sigma_1 + (1-t)\Sigma_2.
\end{equation}

Second, $\mathcal S$ is the convex hull of the set of orthogonal projection matrices, thus $\mathcal S_\nu = \mathcal S$ if and only if all orthogonal projection matrices belong to $\mathcal S_\nu$. Assume that $P\in\mathcal S_\nu$ and let $\hat x$ be the estimate corresponding to a message generating $P$ as a covariance. Then the covariance of $x-\hat x$ is $I_n-P$ and so (almost surely)
\begin{align}
	x-\hat x &\in \Ima (I_n-P) = \ker P \\
	\hat x &\in \Ima P = \ker (I_n-P).
\end{align}
In turn,
\begin{equation}
	P(x-\hat x) = (I_n-P)\hat x = 0,
\end{equation}
that is,
\begin{equation}
	\hat x = Px.
\end{equation}

As a result, the message $y=\hat x$ is credible in the sense that $\mathbb E[x\mid\hat x] = \hat x$. Conversely, if this message is credible, its estimator is $\hat x$ itself, of covariance $P$. Therefore, for $P$ orthogonal projection matrix, $P\in\mathcal S_\nu$ if and only if (iv) holds for this specific $P$. All in all, (i) and (iv) are equivalent. 

\medskip

It is clear that (i) implies (ii). Assume (ii) holds, then for any $P$ orthogonal projection matrix 
\begin{equation}
	\min_{\Sigma\in \mathcal S_\nu} ~\Tr((I_n-2P)\Sigma) 
	= \min_{\Sigma\in \mathcal S} ~\Tr((I_n-2P)\Sigma) = \Tr(-P),
\end{equation}
using property (ii) with $D=I_n-2P$. Thanks to \Cref{lem:trsols}, the only matrix $X\in\mathcal S \supset \mathcal S_\nu$ solution of the second program is $P$ itself, therefore it must be that $P\in\mathcal S_\nu$, hence $(i)$ stands by convexity of $\mathcal S_\nu$.
%
%
\end{proof}

The last piece of the puzzle is to establish the equivalence between (iii) and any of the other conditions of \Cref{lem:whichcov}. Condition (iv) proves instrumental in this endeavor since at the heart it really states that the Radon transform of $\nu$ is rotationally-invariant. The use of the Radon transform here is similar in spirit to \cite{richards1986positive}, studying $\alpha$-symmetric distributions.

\begin{lemma}
\label{lem:isotrop}
	The following statements are equivalent,
	\begin{enumerate}
		\item[(iii)] for all rotation matrix $R$, $Rx \sim \nu$;
		\item[(iv)] for all orthogonal projection matrix $P$, $\mathbb E[x\!\mid\!Px] = Px$.
	\end{enumerate}
\end{lemma}

\begin{proof}[Proof of \Cref{lem:isotrop}]
When $n=1$, both statements are vacuously true since $R=I_1$ is the only rotation, and since $P=0$ is the only projection matrix of rank $0$ and $\nu$ is centered. We thus assume in the remainder of the proof that $n\geq2$. 

\medskip

We first prove that (iii) imply $-x\sim\nu$, then that this implies (iv). Assume (iii) holds and take a Euclidean ball $B$, call $x_0$ its center. The opposite ball, $-B$, is simply $RB$ where $R$ is any rotation that maps $x_0$ to $-x_0$. This rotation exists precisely because rotations act transitively on $\mathbb R^n$, since $n\geq2$. As a result, $(-I_n)_*\nu$ and $(R^{-1})_*\nu = \nu$\footnote{When $f\colon \mathbb R^n \to (E,\Sigma)$ is a measurable function, we let $f_*\nu$ denote the pushforward measure on $(E,\Sigma)$ defined by $f_*\nu(A) = \nu(f^{-1}(A))$ for all $A\in\Sigma$. When $M$ is an $n\times n$-matrix, we identify it with the endomorphism of $\mathbb R^n$, $x\mapsto Mx$, by slight abuse of notation.} agree on all balls, hence are equal \cite{preiss1991measures}. This establishes that $\nu$ is orthogonally-invariant, i.e., isotropic, and therefore the distribution of $x_1,\dots,x_i$ conditional on $x_{i+1},\dots,x_n$ is isotropic \cite{isotropic}, hence centered and so
\begin{equation}
	\mathbb E[x_1,\dots,x_i \!\mid\! x_{i+1},\dots,x_n] = 0.
\end{equation}
Given that $\nu$ is isotropic, (iv) holds.

\medskip

We now establish the converse direction which relies on a Radon transform of sorts. Assume (iv) holds. We first define a proxy for the density of $\nu$, $f_\varphi$, then show its Radon transform is radial, which we use to prove that the Fourier transform of $f_\varphi$ is radial as well. This last fact is shown to imply that $f_\varphi$ itself is radial, which in turn proves that $\nu$ is isotropic.


Since $\nu$ is not assumed to have a density with respect to the Lebesgue measure and since it is not assumed compactly-supported, we may not define the Radon transform in any traditional way. Instead, we define a function $f_\varphi$ which serves as a proxy for the density, where $\varphi$ is a bump function. For a (compactly-supported smooth) bump function $\varphi\in\mathcal D(\mathbb R) = \mathcal C^\infty_c(\mathbb R)$ then, we define the smooth and integrable function $f_\varphi \in \mathcal C^\infty(\mathbb R^n) \cap L^1(\mathbb R^n)$ by
	\begin{equation}
		f_\varphi(x) = \int_{\mathbb R^n} \varphi(\|x-y\|) \,\mathrm d\nu(y).
	\end{equation}

Let us recall how the Radon transform is defined for functions. When $\omega \in \mathbb S^{n-1}$ and $p \in \mathbb R$, we understand the couple $(\omega,p)\in\mathbb P^n$ as the affine hyperplane $p\omega + \omega^\perp$, noting that $(-\omega,-p) = (\omega,p)$ so that $\mathbb S^{n-1}\times\mathbb R$ is a double cover of $\mathbb P^n$ (the space of affine hyperplanes). The Radon transform of a function $f\in L^1(\mathbb R^n)$ is denoted by $\hat f$ and defined as
\[
	\hat f(\omega,p) = \int_{(\omega,p)} f(x) \,\mathrm dx,
\]
where $\mathrm dx$ is the Euclidean measure on the affine hyperplane $(\omega,p)$. 

Let us apply this transformation to  $f_\varphi$. Proceeding to the change of variable $x = y + (p-\omega^\top y)\omega + h$ with $h\in\omega^\perp$, then $h=Rk$ where $R$ is a rotation such that $\omega=Re_1$, we obtain
	\begin{align}
		\hat f_\varphi(\omega,p)
		&= \int_{\mathbb R^n} \int_{(\omega,p)} \varphi(\|x-y\|) \,\mathrm dx \,\mathrm d\nu(y) \\
		&= \int_{\mathbb R^n} \int_{\omega^\perp} \varphi(\|(p-\omega^\top y)\omega + h\|) \,\mathrm dh \,\mathrm d\nu(y) \\
		&= \int_{\mathbb R^n} \int_{e_1^\perp} \varphi(\|(p-\omega^\top y)e_1 + k\|) \,\mathrm dk \,\mathrm d\nu(y),
	\end{align}
	where $\mathrm dx,\mathrm dh,\mathrm dk$ are the Euclidean measures on respectively $(\omega,p)$, $\omega^\perp$ and $e_1^\perp$. We note that the inner integral is a compactly-supported smooth function of $\omega^\top y$. The following lemma implies then that the Radon transform of $f_\varphi$ is radial (i.e., $\hat f_\varphi(\omega,p)$ only depends on $p$, and not on $\omega$).
	

	\begin{lemma}
		Assume that $n\geq2$ and (iv) holds, namely that for all orthogonal projection matrix $P$, $\mathbb E[x\!\mid\!Px] = Px$. Then for all $\phi\in\mathcal D(\mathbb R) = \mathcal C_c^\infty(\mathbb R)$, the following function is constant, 
		\[
			\omega\in\mathbb S^{n-1} \longmapsto \int_{\mathbb R^n} \phi(\omega^\top y) \,\mathrm d\nu(y) \in \mathbb R.
		\]
	\end{lemma}

	
		\begin{proof}
			Fix $\phi$ and 
			let $\gamma \colon I \to \mathbb S^{n-1}$ be a continuously differentiable path with $I \subset \mathbb R$ open. For all $t\in I$, define
	\begin{equation}
		\mathcal I(t) = \int_{\mathbb R^n} \phi(\gamma(t)^\top x) \,\mathrm d\nu(x).
	\end{equation}
	Surely $\gamma'^\top \gamma=0$, and so thanks to the Leibniz integral rule,
	\begin{align}
		\mathcal I'(t)
		&= \int_{\mathbb R^n} 
		(\gamma'(t)^\top x) \phi'(\gamma(t)^\top x) \,\mathrm d\nu(x) \\
		&= \gamma'(t)^\top \mathbb E[ (x-Px) \phi'(\gamma(t)^\top Px) ] \\
		&= \gamma'(t)^\top\mathbb E[\phi'(\gamma(t)^\top Px)\mathbb E[x-Px\!\mid\! Px]] \\
		&=0,
	\end{align}
	having denoted $P = \gamma(t)\gamma(t)^\top$. This shows that $\mathcal I$ is constant. Any two points on $\mathbb S^{n-1}$ can be joined by a continuously differentiable path, precisely because $n\geq2$, and as a result,
	\begin{equation}
		\int_{\mathbb R^n} \phi(\omega^\top x) \,\mathrm d\nu(x)
	\end{equation}
	does not depend on $\omega\in\mathbb S^{n-1}$.
	\end{proof}

Let us denote the Fourier transform with a tilde here. The (multi-dimensional) Fourier transform of $ f_{\varphi}$ at $p\omega\in\mathbb R^n$ is conveniently expressed using the Radon transform of $f_{\varphi}$:
\begin{align}
	\tilde f_\varphi(p\omega)
	&= \int_{-\infty}^\infty \int_{(\omega,r)} f_\varphi(x) e^{-ip\omega^\top x}\,\mathrm dx \,\mathrm dr \\
	&= \int_{-\infty}^\infty \hat f_\varphi(\omega,r) e^{-ipr}\,\mathrm dr,
\end{align}
having denoted the Euclidean measure on $(\omega,r)$ by $\mathrm dx$. As a result, $\tilde f_\varphi$ is radial. Moreover, $\tilde f_\varphi$ is integrable as,
\begin{align}
	\int_{\mathbb R^n} |\tilde f_\varphi(x)| \,\mathrm dx
	&\leq \int_{\mathbb R^n} \int_{\mathbb R^n} |\varphi(\|x-y\|)| \,\mathrm dy \,\mathrm d\nu(x) \\
	&= \int_{\mathbb R^n} |\varphi(\|y\|)| \,\mathrm dy < \infty.
\end{align}
The next lemma then establishes that $f_\varphi$ itself is radial.

\begin{lemma}
	If $f\in\mathcal C^0(\mathbb R^n) \cap L^1(\mathbb R^n)$ is such that its Fourier transform is absolutely integrable and radial (i.e., that $\tilde f(x)$ only depends on $\|x\|$), then $f$ itself is radial.
\end{lemma}

\begin{proof}
 We define $T$, endomorphism of $L^1(\mathbb R^n)$, by
\begin{equation}
	Th(x) = \int_{SO(n)} h(Rx) \mathrm dR,
\end{equation}
where $\mathrm dR$ is the Haar measure on $SO(n)$ (the space of rotations). Fubini's theorem shows that $\|Th\|_1 \leq \|h\|_1$. The result $Th$ is radial, and so $h=Th$ if and only if $h$ is radial. When $h\in L^1(\mathbb R^n)$, an elementary application of Fubini's theorem shows that $\widetilde{Th} = T\tilde h$. Now $g=f-Tf$ is continuous, absolutely integrable and its Fourier transform is null since $\tilde f$ is radial:
\begin{equation}
	\tilde g = \tilde f - T \tilde f = 0.
\end{equation}
The Fourier inversion theorem \cite{folland2009fourier} implies that $g=0$, thus $f = Tf$, and so $f$ itself is radial. 
\end{proof}

Finally, let $B$ be an open ball, denote $y\in\mathbb R^n$ its center and $r>0$ its radius. Let $R\in SO(n)$ be a rotation. Consider a non-decreasing sequence of bump functions $(\varphi_k)_k \subset \mathcal D(\mathbb R)$ with limit $x\in\mathbb R\mapsto \mathds 1_{|x|<r}$. By the monotone convergence theorem,
\[
	\nu(B) = \lim_{k\to\infty} f_{\varphi_k}(y) = \lim_{k\to\infty} f_{\varphi_k}(Ry) = \nu(RB).
\]
This entails that $\nu(B) = \nu(RB) = (R^{-1})_*\nu(B)$ whenever $R \in SO(n)$ and $B$ is an open ball. Since $\nu$ and $(R^{-1})_*\nu$ are finite Borel measures on $\mathbb R^n$, the result of \cite{preiss1991measures} implies that $\nu = (R^{-1})_*\nu$.
\end{proof}

\end{document}